\documentclass[11pt]{article}
\usepackage{fullpage}
\usepackage{amsmath, amsthm, amssymb, graphicx}
\usepackage{hyperref}
\usepackage{calc}
\usepackage{tikz}
\usepackage{subfigure}
\usepackage{extras}
\usepackage{psfrag}
\usepackage[noend]{algorithmic}
\usepackage{algorithm}
\usepackage{tweaklist}
\renewcommand{\enumhook}{\setlength{\itemsep}{-0.5ex}}

\newtheoremstyle{mytheorem}{\topsep}{\topsep}{\sffamily}{}{\bfseries}{.}{.5em}{}
\theoremstyle{mytheorem}
\newtheorem{theorem}{Theorem}

\newtheorem{lemma}[theorem]{Lemma}
\newtheorem{corollary}[theorem]{Corollary}
\newtheorem{proposition}[theorem]{Proposition}
\theoremstyle{definition}
\newtheorem{definition}[theorem]{Definition}

\DeclareMathOperator{\dist}{dist}

\DeclareMathOperator{\tw}{tw}
\newcommand{\eps}{\epsilon}
\newcommand{\prob}[1]{\textit{#1}}

\newcommand{\DD}{\mathcal{D}}
\newcommand{\eS}{\mathcal{S}}

\newcommand{\C}{\mathcal{C}}
\newcommand{\B}{\mathcal{B}}

\newcommand{\T}{\mathcal{T}}

\newcommand{\algo}[1]{\textsc{#1}}

\newcommand{\OPT}{\mbox{\rm OPT}}

\newcommand{\Group}{\mathcal{G}}
\newcommand{\myell}{p}

\newif\ifappendix
\newif\ifmain
\newif\iffullorapp
\maintrue
\fullorappfalse
\newif\ifabstract
\abstracttrue
 \abstractfalse 
\newif\iffull
\ifabstract \fullfalse \else \fulltrue \fullorapptrue \fi

\newcommand{\dynproof}[2]{\noindent \textit{Proof of \eqref{#1} left
    #2 \eqref{#1} right:} \medskip}
\psfragscanon
\begin{document}
\iffull
 \date{ver.~16}
\fi
\date{}
\title{Approximation Schemes for Steiner Forest on
  Planar Graphs \\ and Graphs of Bounded Treewidth}
\author{MohammadHossein Bateni\thanks{Department of Computer Science, Princeton University, Princeton, NJ 08540; Email: \textsf{mbateni@cs.princeton.edu}.
The author was supported by NSF ITR grants
                      CCF-0205594, CCF-0426582 and NSF CCF 0832797,
                      NSF CAREER award CCF-0237113,
                      MSPA-MCS award 0528414,
                      NSF expeditions award 0832797,
  as well as a Gordon Wu fellowship.}
 \and MohammadTaghi Hajiaghayi\thanks{AT\&T Labs--Research, Florham
   Park, NJ 07932; Email: \textsf{hajiagha@research.att.com}.}
\and
 D\'{a}niel Marx\thanks{School of Computer Science, Tel Aviv
   University, Tel Aviv, Israel, \textsf{dmarx@cs.bme.hu.}}
}
\maketitle
\begin{abstract}
We give the first \emph{polynomial-time approximation scheme} (PTAS)
for the \prob{Steiner forest} problem on planar graphs and, more
generally, on graphs of bounded genus. As a first step, we show how to build a
\emph{Steiner forest spanner} for such graphs. The crux of the process
is a clustering procedure called {\em prize-collecting clustering}
that breaks down the input instance into separate subinstances which are
easier to handle; moreover, the terminals in different subinstances are 
far from each other. Each subinstance has a
relatively inexpensive Steiner tree connecting all its terminals,
and the subinstances can be solved (almost) separately. 
Another building block is a PTAS for \prob{Steiner forest} on graphs of
bounded treewidth.
Surprisingly, \prob{Steiner forest} is NP-hard even on graphs
of treewidth 3. Therefore, our PTAS for bounded treewidth graph needs
a nontrivial combination of approximation arguments and dynamic
programming on the tree decomposition. We further show that \prob{Steiner forest} can be solved
in polynomial time for series-parallel graphs (graphs of treewidth
at most two) by a novel combination of dynamic programming and minimum
cut computations, completing our
thorough complexity study of \prob{Steiner forest} in the range of bounded
treewidth graphs, planar graphs, and bounded genus graphs.
\end{abstract}
\section{Introduction}

One of the most fundamental problems in combinatorial optimization
and network design with both practical and theoretical significance
is the Steiner forest problem, in which given a graph $G=(V,E)$ and
a set consisting of pairs of terminals, called {\em demands}, $\DD=
\{(s_1, t_1), (s_2,t_2),\dots, (s_k,t_k)\}$, the goal is to find a
minimum-cost forest $F$ of $G$ such that every pair of terminals in $\DD$ is
connected by a path in $F$. The first and the best approximation
factor for this problem is 2 due to Agrawal, Klein and
Ravi~\cite{AKR95} (see also Goemans and Williamson~\cite{GW95}). The
conference version of the Agrawal, Klein and Ravi~\cite{AKR91}
appeared in 1991, and there have been no improved approximation
algorithms invented for Steiner forest. Recently Borradaile, Klein
and Kenyon-Mathieu~\cite{BKM08:euc-for} obtain a Polynomial Time
Approximation Scheme (PTAS) for Euclidean Steiner forest where the
terminals are in the Euclidean plane. They pose obtaining a PTAS for
Steiner forest in planar graphs, the natural generalization of
Euclidean Steiner forest, as the main open problem. We note that in
network design, planarity is a natural restriction, since in
practical scenarios of physical networking, with cable or fiber
embedded in the ground, crossings are rare or nonexistent. In this
paper, we settle this open problem by obtaining a PTAS for planar
graphs (and more generally, for bounded genus graphs)
via a novel technique of {\em prize-collecting clustering} with
potential use to obtain other PTASs in planar graphs.

The special case of the Steiner forest problem when all pairs have a common terminal is
the classical Steiner tree problem, one of the first problems shown NP-hard by
Karp~\cite{cr:18}. The problem remains hard even on planar
graphs~\cite{GJ77}. In contrast to Steiner forest, a long sequence
of papers give approximation factors better than 2 for this
problem~\cite{ZelAnnals,Zelikovsky93-Steiner,BermanR92,Zelik96,PromelS97,KarpinskiZ97,HP99,RobinsZ05};
the current best approximation ratio is 1.55~\cite{RobinsZ05}. Since
the problem is APX-hard in general graphs~\cite{cr:7,Thimm03}, we do
not expect to obtain a PTAS for this problem in general graphs.
However, for the Euclidean Steiner tree problem,
Arora~\cite{arora98:ptas} and Mitchell~\cite{cr:26} present a PTAS.
Obtaining a PTAS for Steiner tree on planar graphs, the natural
generalization of Euclidean Steiner tree, remained a major open
problem since the conference version of Arora~\cite{A96} in 1996.
Only in 2007, Borradaile, Kenyon-Mathieu, and
Klein~\cite{BKM07:planar} settle this problem with a nice technique
of constructing \emph{light spanners} for Seiner trees in planar graphs. In
this paper, we also generalize this result to obtain light spanners
for Steiner forests.

Most approximation schemes for planar graph problems use (implicitly
or explicitly) the fact that the problem is easy to solve on
bounded-treewidth graphs. In particular, a keystone blackbox in the
algorithm of Borradaile et al.~\cite{BKM07:planar} for Steiner tree is
the result that, for every fixed value of $k$, the problem is
polynomial-time solvable on graphs having treewidth at most $k$. There
is a vast literature on algorithms for bounded-treewidth graphs and in
most cases polynomial-time (or even linear-time) solvability follows
from the well-understood standard technique of dynamic programming on
tree decompositions.  However, for Steiner {\em forest,} the obvious
way of using dynamic programming does not give a polynomial-time
algorithm.  The difficulty is that, unlike in Steiner tree, a solution
of Steiner forest induces a partition on the set of terminals and a
dynamic programming algorithm needs to keep track of exponentially
many such partitions. In fact, this approach seems to fail even for
series-parallel graphs (which have treewidth at most 2); the
complexity of the problem for series-parallel graphs was stated as an
open question by Richey and Parker~\cite{MR862895} in 1986. We resolve
this question by giving a polynomial-time algorithm for Steiner forest
on series-parallel graphs.  The main idea is that even though
algorithms based on dynamic programming have to evaluate subproblems
corresponding to exponentially many partitions, the function
describing these exponentially many values turn out to be submodular
and it can be represented in a compact way by the cut function of a
directed graph. On the other hand, Steiner forest becomes NP-hard on
graphs of treewidth at most 3. Thus perhaps this is the first example
when the complexity of a natural problem changes as treewidth
increases from 2 to 3.  In light of this hardness result, we
investigate the approximability of the problem on bounded-treewidth
graphs and show that, for every fixed $k$, Steiner forest admits a
PTAS on graphs of treewidth at most $k$. The main idea of the PTAS is
that if the dynamic programming algorithm considers only an
appropriately constructed polynomial-size subset of the set of all
partitions, then this produces a solution close to the optimum. Very
roughly, the partitions in this subset are constructed by choosing a
set of center points and classifying the terminals according to the
distance to the center points.  Our PTAS for planar graphs (and more
generally, for bounded genus graphs) uses this PTAS for
bounded-treewidth graphs.  This completes our thorough study of
Steiner forest in the range of bounded treewidth graphs, planar graphs
and bounded genus graphs.

\subsection{Our results and techniques}
 Our main result in this paper is a PTAS for the planar Steiner
forest problem.
\begin{theorem}\label{thm:main}
For any constant $\bar{\eps} > 0$, there is a polynomial-time
$(1+\bar{\eps})$-approximation algorithm for the Steiner forest
problem on planar graphs and, more generally, on graphs of bounded genus.
\end{theorem}

 To this end, we build a \emph{Steiner forest spanner} for the
input graph and the set of demands in Section~\ref{sec:spanner}.
Roughly speaking, a Steiner forest spanner is a subgraph of the
given graph whose cost is no more than a constant factor times the
cost of the optimal Steiner forest, and furthermore, it contains a
nearly optimal Steiner forest. Denote by $\OPT_{\DD}(G)$ the minimum
cost of a Steiner forest of $G$ satisfying (connecting) all the
demands in $\DD$. We sometimes use $\OPT$ instead of
$\OPT_{\DD}(G)$. A subgraph $H$ of $G$ is a \emph{Steiner forest
spanner} with respect to demand set $\DD$ if it has the following
two properties:
\begin{description}
 \item[Spanning Property:] There is a forest in $H$ that connects all demands in $\DD$ and has
length at most $(1 + \eps)OPT_{\DD}(G)$, namely, $\OPT_{\DD}(H) \leq
(1+\eps)\OPT_{\DD}(G)$.
 \item[Shortness Property:] The total length of $H$ is at most $f(\eps)OPT_{\DD}(G)$.
\end{description}

\begin{theorem}\label{thm:spanner}
 Given any $\eps > 0$,
a bounded genus graph $G_{in}(V_{in},E_{in})$ and demand pairs $\DD$, we can
compute in polynomial time a Steiner forest spanner $H$ for $G_{in}$
with respect to demand set $\DD$.
\end{theorem}
The algorithm that we propose achieves this in time $O(n^2\log n)$.
The entire algorithm for Steiner forest runs in polynomial time but the exponent of the polynomial depends on $\epsilon$ and the genus of the input graph.

The proof of Theorem~\ref{thm:spanner} heavily relies on a novel
clustering method presented in Theorem~\ref{thm:break} that allows
us to (almost) separately build the spanners for smaller and far
apart sets of demands. Construction of the spanner for each of the
sets itself uses  ideas of Borradaile et al.~\cite{BKM07:planar},
although there are still several technical differences.
The clustering technique works for general graphs as opposed to
the rest of the construction which requires the graphs to
have bounded genus.
\begin{theorem}\label{thm:break}
 Given an  $\eps > 0$,
a graph $G_{in}(V_{in},E_{in})$, and a set $\DD$ of pairs of
vertices, we can compute in polynomial time a set of trees $\{T_1,
\dots, T_k\}$, and a partition of demands $\{\DD_1, \dots, \DD_k\}$,
 with the following properties.
\begin{enumerate}
 \item
   All the demands are covered, i.e., $\DD=\bigcup_{i=1}^k \DD_i$.
 \item
   All the terminals in $\DD_i$ are spanned by the tree $T_i$.
 \item
   The sum of the costs of all the trees $T_i$ is no more than $(\frac{4}{\eps}+2)\OPT$.
 \item
   The sum of the costs of  minimum Steiner forests of all demand sets $\DD_i$ is no more than $1+\eps$ times the cost of a minimum Steiner forest of $G_{in}$; i.e., $\sum_i\OPT_{\DD_i}(G_{in}) \leq (1+\eps)\OPT_{\DD}(G_{in})$. 
\end{enumerate}
\end{theorem}
The last condition implies that (up to a small factor) it is possible to solve
the demands $\DD_i$ separately.
Notice that this may lead to paying for poritions of the solution more than once.

We will  prove Theorem~\ref{thm:break} in Section~\ref{sec:break}.
Roughly speaking, the algorithm here first recognizes some connected
components by running a 2-approximation Steiner forest algorithm.
Clearly, this construction satisfies all but the last condition of the theorem.
However, at this point, these connected components might not be sufficiently far
from each other such that we can consider them separately. To fix this, we
contract each connected component into a ``super vertex'' to which we
assign a prize (potential) which is proportional to the sum of the
edge weights of the corresponding component. Now we are running an
algorithm that we call {\em prize-collecting clustering}. The
algorithm as well as some parts of its analysis bears similarities
to a primal-dual method due to Agrawal, Klein and Ravi~\cite{AKR91}
and Goemans and Williamson~\cite{GW95}. Indeed our analysis
strengthens these previous approaches by proving some local
guarantees (instead of the global guarantee provided in these
algorithms). In some sense, this clustering algorithm can also be
seen as a generalization of an implicit clustering algorithm of
Archer, Bateni, Hajiaghayi, and Karloff~\cite{ABHK09} who improve
the best approximation factor for prize-collecting Steiner tree to
$2-\eps$, for some constant $\eps> 0$. In this clustering, we
consider a topological structure of the graph in which each edge is
a curve connecting its endpoints whose length is equal to its
weight. We color (portions of) edges by different colors each
corresponding to a super vertex. These colors form a laminar family
and the ``depth'' of each color is at most the prize given to its
corresponding super vertex. Using this coloring scheme we further
connect some of these super vertices to each other with a cost
proportional to the sum of their prizes. At the end, we show that
now we can consider these combined connected components as separate
clusters and, roughly speaking, an optimum solution need not connect two different clusters
because of the concept of depth. We believe the prize-collecting
clustering presented in this paper might have applications for
other problems (esp. to obtain PTASs).

To obtain a PTAS promised in Theorem~\ref{thm:main}, first we
construct a spanner Steiner forest provided in
Theorem~\ref{thm:spanner}. On this spanner, we utilize a technique
due to Klein~\cite{Klein08:tsp}, and Demaine, Hajiaghayi and
Mohar~\cite{DHM07} which reduces the problem of obtaining a PTAS in
a planar (and more generally, bounded genus) graph whose total edge
weight is within a constant factor of the optimum solution to that
of finding an optimal solution in a graph of bounded treewidth (see
the proof of Theorem~\ref{thm:main} in Section~\ref{sec:main-ptas}
for more details.) However, there are no known
polynomial-time algorithms in the literature for Steiner forest on
bounded-treewidth graphs, so we cannot plug in such an algorithm to
complete our PTAS. Therefore, we need to investigate Steiner forest on
bounded-treewidth graphs.

Resolving an open question of Richey and Parker~\cite{MR862895} from
1986, we design a polynomial-time algorithm for Steiner forest on series-parallel
graphs. Very recently, using completely different techniques, a
polynomial-time algorithm was presented for the special case of
outerplanar graphs~\cite{Gassner2009}.

\begin{theorem}\label{th:mainalg}
The Steiner forest problem can be solved in polynomial time for
subgraphs of series-parallel graphs (i.e., graphs
of treewidth at most 2).
\end{theorem}

\iffull
A series-parallel graph can be built form elementary blocks using two
operations: parallel connection and series connection.
The algorithm of Theorem~\ref{th:mainalg} uses dynamic programming on
the construction of the series-parallel graph. For each subgraph
arising in the construction, we find a minimum weight forest that
connects some of the terminal pairs, connects a subset of the
terminals to the ``left exit point'' of the subgraph, and connects the
remaining terminals to the ``right exit point'' of the
subgraph. The minimum weight depends on the subset of terminals connected to the left
exit point, thus it seems that we need to determine exponential many
values (one for each subset). Fortunately, it turns out that
the minimum weight is a submodular function of the
subset. Furthermore, we show that this function can be represented by
the cut function of a directed graph and this directed graph can be
easily constructed if the directed graphs corresponding to the
building blocks of the series-parallel subgraph are available. Thus,
following the construction of the series-parallel graph, we can build
all these directed graphs and determine the value of the optimum
solution by the computation of a minimum cut.
\fi

Surprisingly, it turns out that the problem becomes NP-hard on graphs
of treewidth at most 3~\cite{Gassner2009}. For completeness, in
Section~\ref{sec:hardness-treewidth-3} we give a (different) NP-hardness
proof, which highlights how submodularity (and hence the approach of
Theorem~\ref{th:mainalg}) breaks if treewidth is 3. There exist a few known
problems that are polynomial-time solvable for trees but NP-hard for
graphs of treewidth 2 \cite{MR2000i:68160,
  marx-planar-edge-prext,marx-sum-edge-hard,MR862895}, but to our
knowledge, this is the first natural example where there is a
complexity difference between treewidth 2 and 3.

Not being able to solve Steiner forest optimally on
bounded-treewidth graphs is not an unavoidable obstacle for obtaining a PTAS on planar graphs: the
technique of Klein, and Demaine, Hajiaghayi and Mohar also can be
applied when we have a PTAS for 
graphs of bounded treewidth. In
Section~\ref{sec:ptas-bound-treew}, we demonstrate such a PTAS:

\begin{theorem}\label{th:boundedtwptas}
For every fixed $w\ge 1$ and $\epsilon>0$, there is a
polynomial-time $(1+\epsilon)$-approximation algorithm for
\prob{Steiner Forest} on graphs with treewidth at most $w$.
\end{theorem}

Note that the exponent of the polynomial in
Theorem~\ref{th:boundedtwptas} depends on both $\epsilon$ and $w$; it
remains an interesting question for future research whether this
dependence can be removed.

The main idea of the PTAS of Theorem~\ref{th:boundedtwptas} is to
reduce the set of partitions considered in the dynamic programming
algorithm to a polynomially bounded subset, in a way that an
$(1+\epsilon)$-approximate solution using only these partitions is
guaranteed to exist. The implementation of this idea consists of three
components. First, we have to define which partitions belong to the
polynomially bounded subset. These partitions are defined by choosing
a bounded number of center points and a radius for each center. A
terminal is classified into a class of the partition based which
center points cover it. Second, we need an algorithm that finds the
best solution using only the allowed subset of partitions. This can be
done following the standard dynamic programming paradigm, but the
technical details are somewhat tedious. Third, we have to argue that
there is a $(1+\epsilon)$-approximate solution using only the allowed
partitions. We show this by proving that if there is a solution that
uses partitions that are not allowed, then it can be modified,
incurring only a small increase in the cost, such that it uses only
allowed partitions. The main argument here is that for each partition
appearing in the solution, we try to select suitable center points. If
these center points do not generate the required partition, then this
means that a terminal is misclassified, which is only possible if the
terminal is close to a center point. In this case, we observe that two
components of the solution are close to each other and we can join
them with only a small increase of the cost. The crucial point of the
proof is a delicate charging argument, making use of the structure of
bounded-treewidth graphs, which shows that repeated applications of
this step results in a total increase that is not too large.

\section{Basic definitions}\label{sec:defs}
Let $G(V,E)$ be a graph. As is customary, let $\delta(V')$ denote
the set of edges having one endpoint in a subset $V'\subseteq V$ of
vertices. For a subset of vertices $V'\subseteq V$, the subgraph of
$G$ induced by $V'$ is denoted by $G[V']$. With slight abuse of
notation, we sometimes use the edge set to refer to the graph
itself. Hence, the above-mentioned subgraph may also be referred to
by $E[V']$ for simplicity. We denote the length of a shortest
$x$-to-$y$ path in $G$ as $\dist_G(x, y)$. For an edge set $E$, we
denote by $\ell(E):=\sum_{e\in E}c_e$ the total length of edges in
$E$. 

A collection $\eS$ is said to be \emph{laminar} if and only if
for any two sets $C_1, C_2\in\eS$, we have $C_1\subseteq C_2$,
 $C_2\subseteq C_1$, or $C_1\cap C_2=\emptyset$.
Suppose $\C$ is a partition of a ground set $V$. Then, $\C(v)$
denotes for each $v\in V$ the set $C\in\C$ that contains $v$.

Given an edge $e=(u,v)$ in a graph $G$, the \emph{contraction} of
$e$ in $G$ denoted by $G/e$ is the result of unifying vertices $v$
and $w$ in $G$, and removing all loops and multiple edges except the
shortest edge. \iffull More formally, the contracted graph $G/e$ is
formed by the replacement of $u$ and $v$ with a
 single vertex such that edges incident to the new vertex are the edges other than $e$ that
  were incident with $u$ or $v$. To obtain a simple graph, we first remove all self-loops
  in the resulting graph.
In case of multiple edges, we only keep the shortest edge and remove
all the rest.\fi The contraction $G/E'$ is defined as the result of
iteratively contracting all the edges of $E'$ in $G$, i.e., $G/E' :=
G/e_1/e_2/\dots/e_k$ if $E'=\{e_1, e_2, \dots, e_k\}$. Clearly, the
planarity of $G$ is preserved after the contraction. Similarly,
contracting edges does not increase the cost of an optimal Steiner
forest.

   The boundary of a face of a planar embedded graph is the set of edges adjacent to the face; it does not always form a simple cycle. The boundary
$\partial H$ of a planar embedded graph H is the set of edges bounding the infinite face. An edge is strictly enclosed by the boundary of $H$ if the edge belongs to $H$ but not to $\partial H$.

\iffull Now we define the basic notion of treewidth, as introduced
by Robertson and Seymour~\cite{RS86}.  To define this notion, we
consider representing a graph by a tree structure, called a tree
decomposition. More precisely, a \emph{tree decomposition} of a
graph $G(V,E)$ is a pair $(T,\B)$ in which $T(I,F)$ is a tree and
$\B=\{B_i\:|\:i\in I\}$ is a family of subsets of $V(G)$ such that
1) $\bigcup_{i\in I}B_i = V$; 2) for each edge $e=(u,v)\in E$, there
exists an $i\in I$ such that both $u$ and $v$ belong to $B_i$; and
3) for every $v\in V$, the set of nodes $\{i\in I\:|\:v\in B_i\}$
forms a connected subtree of $T$.

To distinguish between vertices of the original graph $G$ and
vertices of $T$ in the tree decomposition, we call vertices of $T$
\emph{nodes} and their corresponding $B_i$'s bags.  The \emph{width}
of the tree decomposition is the maximum size of a bag in $\B$ minus
$1$.  The \emph{treewidth} of a graph $G$, denoted $\tw(G)$, is the
minimum width over all possible tree decompositions of $G$.

For algorithmic purposes, it is convenient to define a restricted
form of tree decomposition. We say that a tree decomposition
$(T,\B)$  is {\em nice} if the tree $T$ is a rooted tree such that
for every $i\in I$ either
\begin{enumerate}
\item $i$ has no children ($i$ is a {\em leaf node}),
\item $i$ has exactly two children $i_1$, $i_2$ and
  $B_i=B_{i_1}=B_{i_2}$ holds ($i$ is a {\em join node}),
\item $i$ has a single child $i'$ and $B_i=B_{i'}\cup\{v\}$ for
  some $v\in V$ ($i$ is an {\em introduce node}), or
\item $i$ has a single child $i'$ and $B_i=B_{i'}\setminus \{v\}$ for
  some $v\in V$ ($i$ is a {\em forget node}).
\end{enumerate}
It is well-known that every tree decomposition can be transformed
into a nice tree decomposition of the same width in polynomial time.
Furthermore, we can assume that the root bag contains only a single vertex.

We will use the following lemma to obtain a nice tree decomposition
with some further properties (a related trick was used in
\cite{marx-gt04}, the proof is similar):
\begin{lemma}\label{lem:nicer}
Let $G$ be a graph having no adjacent degree 1 vertices. $G$ has a
nice tree decomposition of polynomial size having the following two
additional properties:
\begin{enumerate}
\item No introduce node introduces a degree 1 vertex.
\item The vertices in a join node have degree greater than 1.
\end{enumerate}
\end{lemma}
\begin{proof}
  Consider a nice tree decomposition of graph $G$.  First, if $v$ is a
  vertex of degree 1, then we can assume that $v$ appears only in one
  bag: if $w$ is the unique neighbor of $v$, then it is sufficient
  that $v$ appears in any one of the bags that contain $w$. Let
  $B_v=\{v,x_1,\dots,x_t\}$ be this bag where $x_1=w$. We modify the tree
  decomposition as follows. We replace $B_v$ with $B'_v=B_v\setminus
  \{v\}$, insert a bag $B''_v=B_v\setminus \{v\}$ between $B'_v$ and
  its parent, and create a new bag $B^t=B_v\setminus \{v\}$ that is
  the other child of $B''_v$ (thus $B''_v$ is a join node). For
  $i=1,\dots,t-1$, let $B^i=\{x_1,\dots,x_i\}$, and let $B^{i}$ be the child of $B^{i+1}$.
  Finally, let $B_w=\{w,v\}$ be the child of $B^1$ and let $B=\{v\}$
  be the child of $B_w$. Observe that $B^i$ ($2\le i \le
  t$), $B_w$ are introduce nodes, $B^1$ is a forget node, and $B$ is
  a leaf node. This operation ensures that vertex $v$ appears only in a leaf
  node. It is clear that after repeating this operation for every
  vertex of degree 1, the two required properties will hold.
\end{proof}

We also need a basic notion of embedding; see, e.g., \cite{RS94,
CM05}. In this paper, an \emph{embedding} refers to a \emph{$2$-cell
embedding}, i.e., a drawing of the vertices and edges of the graph
as points and arcs in a surface such that every face (connected
component obtained after removing edges and vertices of the embedded
graph) is homeomorphic to an open disk. We use basic terminology and
notions about embeddings as introduced in \cite{MT01}.  We only
consider compact surfaces without boundary.  Occasionally, we refer
to embeddings in the plane, when we actually mean embeddings in the
$2$-sphere.  If $S$ is a surface, then for a graph $G$ that is
($2$-cell) embedded in $S$ with $f$ facial walks, the number
$g=2-|V(G)|+|E(G)|-f$ is independent of $G$ and is called the
\emph{Euler genus} of $S$. The Euler genus coincides with the
crosscap number if $S$ is non-orientable, and equals twice the usual
genus if the surface $S$ is orientable.
\fi

\section{Prize-collecting clustering}\label{sec:break}
In this section, we describe an algorithm \algo{PC-Clustering} that is
used to prove Theorem~\ref{thm:break}. The algorithm as well as its
analysis bears similarities to the primal-dual method due to Agrawal,
Klein and Ravi~\cite{AKR95} and Goemans and Williamson~\cite{GW95}. It
uses a technique that we call \emph{prize-collecting clustering} and
our analysis strengthens the previous approaches by proving some local
guarantees (instead of the global guarantee provided in previous
algorithms). We build a forest $F_2$ each of whose components
correspond to one $T_i$ sought in Theorem~\ref{thm:break}. Along the
way, we also come up with a vector $y$ satisfying the sets of
constraints~\eqref{eqn:lp:1}-\eqref{eqn:lp:3} below. During the process, we
maintain a vector $y$ satisfying all these constraints, and at the
end, it will be true that all the constraints~\eqref{eqn:lp:2} hold
with equality.  The analysis takes advantage of the connection between
$F_2$ and $y$.

We start with a $2$-approximate\footnote{Such a solution can be  found via Goemans-Williamson's Steiner forest algorithm, for instance.} solution $F^\ast$ satisfying all the demands in $\DD$;  the cost of $F^\ast$ is at most $2\OPT$.
The forest $F^\ast$ consists of tree components $T^\ast_i$.
In the following, we connect some of these components to make the trees $T_i$.
It is easy to see this construction  guarantees the first two conditions of Theorem~\ref{thm:break}.
We work on a graph $G(V,E)$ formed from $G_{in}$ by contracting each tree component of $F^\ast$.
A potential $\phi_v$ is associated with each vertex $v$ of $G$,
which is $\frac{1}{\eps}$ times the cost of the tree component corresponding to $v$ in case $v$ is the contraction of a tree component, and zero otherwise.

\begin{lp}
& \sum_{S: e\in\delta(S)}\:\sum_{v\in S}y_{S,v} \leq c_e &\hspace{3cm}&&\forall e\in E, \label{eqn:lp:1}\\
& \sum_{S\ni v}y_{S,v} \leq \phi_v &&&\forall v\in V,  \label{eqn:lp:2}\\
& y_{S,v} \geq 0 &&&\forall v\in S\subseteq V. \label{eqn:lp:3}
\end{lp}

These constraints are very similar to the dual LP for the \prob{prize-collecting Steiner tree} problem when $\phi_v$ are thought of as penalty values corresponding to the vertices.
In the standard linear program for the prize-collecting Steiner tree problem,
there is a special \emph{root} vertex to which all the terminals are  to be connected. Then, no set containing the root appears in the formulation.

The solution is built up in two stages.
First we perform an \emph{unrooted growth} to find a forest $F_1$ and a corresponding $y$ vector.
In the second stage, we \emph{prune} some of the edges of $F_1$ to get another forest $F_2$.
Uncontracting the trees $T^\ast_i$  turns $F_2$ into the Steiner trees $T_i$ in the statement of Theorem~\ref{thm:break}.
Below we describe the two phases of Algorithm~\ref{alg:break} (\algo{PC-Clustering}).

\paragraph{Growth}
 We begin with a zero vector $y$, and an empty set $F_1$.
 We maintain a partition $\C$ of vertices $V$ into clusters; it initially consists of singleton sets.
 Each cluster is either \emph{active} or \emph{inactive};
  the cluster $C\in\C$ is \emph{active} if and only if $\sum_{C'\subseteq C}\sum_{v\in C'}y_{C',v} < \sum_{v\in C} \phi_v$.
 A vertex $v$ is \emph{live} if and only if $\sum_{C\ni v} y_{C,v} < \phi_v$.
 Equivalently, a cluster $C\in\C$ is active if and only if there is a live vertex $v\in C$.
 We simultaneously \emph{grow} all the active clusters by $\eta$.
 In particular, if there are $\kappa(C)>0$ live vertices in an active cluster $C$, we increase $y_{C,v}$ by $\eta/\kappa(C)$ for each live vertex $v\in C$.
 Hence, $y_{C}$ defined as $\sum_{v\in C}y_{C,v}$ is increased by $\eta$ for an active cluster $C$.
 We pick the largest value for $\eta$ that does not violate any of the constraints in \eqref{eqn:lp:1} or \eqref{eqn:lp:2}.
 Obviously, $\eta$ is finite in each iteration because the values of these variables cannot be larger than $\sum_v\phi_v$.
 Hence, at least one such constraint goes tight after the growth step.
 If this happens for an edge constraint for $e=(u,v)$,
 then there are two clusters $C_u\ni u$ and $C_v\ni v$ in $\C$;
 at least one of the two is growing.
 We merge the two clusters into $C = C_u \cup C_v$ by adding the edge $e$ to $F_1$, remove the old clusters and add the new one to $\C$.
 Nothing needs to be done if a constraint \eqref{eqn:lp:2} becomes tight.
The number of iterations is at most $2|V|$, because at each event either a vertex dies, or the size of $\C$ decreases.

We can think of the growth stage as a process that
colors portions of the edges of the graph.
This gives a  better intuition to the algorithm,
and makes several of the following lemmas intuitively simple.
Consider a topological structure in which vertices of the graph
are represented by points, and each edge is a curve connecting its endpoints whose length is equal to the weight of the edge.
Suppose a cluster $C$ is growing by an amount $\eta$.
This is distributed among all the live vertices $v\in C$,
where $y_{C,v}$ is increased by  $\eta':=\eta/\kappa(C)$.
As a result, we color by color $v$ a connected portion
with length $\eta'$ of all the edges in $\delta(C)$.
Finally, each edge $e$ gets exactly $\sum_{C: e\in\delta(C)}y_{C,v}$ units
of color $v$.  We can perform a clean-up process, such that all the
portions of color $v$ are consecutive on an edge.\footnote{
We can do without the clean-up if we perform the coloring in a lazy manner.
That is, we do not do the actual color assignment until the edge goes tight or the algorithm terminates.  At this point, we go about putting colors on the edges, and we make sure the color corresponding to any pair $(S,v)$ forms a consecutive portition of the edge.  This property is not needed as part of our algorithm, though, and is merely for the sake of having a nice coloring which is of independent interest.}
Hence, as a cluster expands, it colors its boundary by the amount of growth.
At the time when two clusters merge, their  colors barely touch each other.
At each point in time, the colors associated with the vertices of a cluster form a connected region.

\paragraph{Pruning}
Let $\eS$ denote the set of all clusters formed during the execution
of the growth step.  It can be easily observed that the clusters $\eS$
are laminar and the maximal clusters are the clusters of $\C$. 
In addition, notice that $F_1[C]$ is connected for each $C\in\eS$.
 
 Let
$\B\subseteq\eS$ be the set of all such clusters that are tight,
namely, for each $S\in\B$, we have $\sum_{S'\subseteq S}\sum_{v\in
  S'}y_{S',v} = \sum_{v\in S}\phi_v$.  In the pruning stage, we
iteratively remove some edges from $F_1$ to obtain $F_2$.  More
specifically, we first initialize $F_2$ with $F_1$.  Then, as long as
there is a cluster $S\in\B$ 
such that $F_2\cap\delta(S) = \{e\}$, we remove the edge $e$ from
$F_2$. 

A cluster $C$ is called a \emph{pruned cluster} if it is pruned in the
second stage in which case, $\delta(C)\cap F_2 = \emptyset$.  Hence, a
pruned cluster cannot have non-empty and proper intersection with a
connected component of $F_2$.

\begin{algorithm}
\caption{\algo{PC-Clustering}\label{alg:break}}
\textbf{Input:} planar graph $G_{in}(V_{in}, E_{in})$, and set of demands $\DD$.\\
\textbf{Output:} set of trees $T_i$ with associated $\DD_i$.
\begin{algorithmic}[1]
 \STATE Use the algorithm of Goemans and Williamson~\cite{GW95} to find a $2$-approximate Steiner forest $F^\ast$ of $\DD$,  consisting of tree components $T^\ast_1, \dots, T^\ast_k$.
 \STATE Contract each tree $T^\ast_i$ to build a new graph $G(V,E)$.
 \STATE For any $v\in V$, let $\phi_v$ be $\frac{1}{\eps}$ times the cost of the tree $T^\ast_i$ corresponding to $v$, and zero if there is no such tree.
 \STATE Let $F_1\gets\emptyset$.
 \STATE Let $y_{S,v}\gets 0$ for any $v\in S\subseteq V$.
 \STATE Let $\eS\gets\C\gets\left\{ \{v\}: v\in V\right\}$.
\WHILE {there is a live vertex}
 \STATE Let $\eta$ be the largest possible value such that simultaneously increasing  $y_{C}$ by $\eta$ for all active clusters $C$ does not violate Constraints~\eqref{eqn:lp:1}-\eqref{eqn:lp:3}.
  \STATE Let $y_{\C(v),v}\gets y_{\C(v),v}+\frac{\eta}{\kappa(\C(v))}$ for all live vertices $v$.
  \IF {$\exists e\in E$ that is tight and connects two clusters}
   \STATE Pick one such edge $e=(u,v)$.
   \STATE Let $F_1\gets F_1\cup \{e\}$.
   \STATE Let $C \gets \C(u)\cup\C(v)$.
   \STATE Let $\C\gets \C\cup \{C\} \setminus \{\C(u), \C(v)\}$.
   \STATE Let $\eS \gets \eS \cup \{C\}$.
  \ENDIF
\ENDWHILE
\STATE Let $F_2 \gets F_1$.
\STATE Let $\B$ be the set of all clusters $S\in\eS$ such that
$\sum_{v\in S}y_{S,v} = \sum_{v\in S}\phi_v$.
\WHILE {$\exists S\in\B$ such that $F_2\cap\delta(S)=\{e\}$ for an edge $e$}
  \STATE Let $F_2 \gets F_2 \setminus \{e\}$.
\ENDWHILE
\STATE Construct $F$ from $F_2$ by uncontracting all the trees $T^\ast_i$.
\STATE Let $F$ consist of tree components $T_i$.
\STATE Output the set of trees $\{T_i\}$, along with $\DD_i := \{ (s,t)\in\DD: s,t\in V(T_i)\}$.
\end{algorithmic}
\end{algorithm}

We first bound the cost of the forest $F_2$.
The following lemma is similar to the analysis of the algorithm in~\cite{GW95}.
However, we do not have a primal LP to give a bound on the dual.
Rather, the upper bound for the cost is the sum of all the potential values $\sum_v \phi_v$.
In addition, we bound the cost of a forest $F_2$ that may have
more than one connected component, whereas the prize-collecting Steiner tree
algorithm of~\cite{GW95} finds a connected graph at the end.
\begin{lemma}\label{lem:break:cost}
 The cost of $F_2$ is at most $2\sum_{v\in V} \phi_v$.
\end{lemma}

\iffull
\begin{proof}
 The strategy is to prove that the cost of this forest is at most $2\sum_{v\in S\subseteq V}y_{S,v} =2\sum_{v\in V}\phi_v$.
The equality follows from Equation~\eqref{eqn:lp:2}---it holds with equality at the end of the algorithm.
Recall that the growth phase has several events corresponding to an edge or set constraint going tight.
We first break apart $y$ variables by epoch.  Let $t_j$ be the
  time at which the $j^{\rm th}$ event point occurs in the growth phase ($0=t_0\leq t_1 \leq t_2 \leq \cdots$), so the $j^{\rm th}$ epoch is the interval of time from $t_{j-1}$ to $t_j$.  For each cluster $C$, let $y_{C}^{(j)}$  be the amount by which $y_C:=\sum_{v\in C}y_{C,v}$ grew during epoch $j$, which is  $t_j-t_{j-1}$ if it was active during this epoch, and zero otherwise.  Thus, $y_C = \sum_j y_{C}^{(j)}$.  Because each edge $e$ of $F_2$
  was added at some point by the growth stage when its edge packing constraint \eqref{eqn:lp:1} became tight, we can exactly apportion the cost
  $c_e$ amongst the collection of clusters $\{C : e\in \delta(C)\}$ whose
  variables ``pay for'' the edge, and can divide this up further
  by epoch.  In other words, $c_e = \sum_j \sum_{C:e\in \delta(C)}
  y_{C}^{(j)}$.  We will now prove that the total edge cost from $F_2$
  that is apportioned to epoch $j$ is at most $2 \sum_{C} y_{C}^{(j)}$.  In other words, during each epoch,
  the total rate at which edges of $F_2$ are paid for by all active
  clusters is at most twice the number of active clusters.
  Summing over the epochs yields the desired conclusion.

  We now analyze an arbitrary epoch $j$.  Let $\C_j$ denote the set
  of clusters that existed during epoch $j$.
Consider the graph $F_2$, and then collapse each cluster
  $C \in \C_j$ into a supernode.  Call the resulting graph $H$.
Although the nodes of
  $H$ are identified with clusters in $\C_j$, we will continue to refer
  to them as clusters, in order to to avoid confusion with the nodes of
  the original graph.  Some of the clusters are active and some may be
  inactive.  Let us denote the active and inactive clusters in $\C_j$
  by $\C_{act}$ and $\C_{dead}$, respectively.
  The edges of $F_2$ that are being partially paid for during epoch $j$
  are exactly those edges of $H$ that are incident to an active cluster,
  and the total amount of these edges that is paid off during epoch
  $j$ is $(t_j-t_{j-1}) \sum_{C \in \C_{act}} \deg_H(C)$.
  Since every active cluster grows by exactly $t_j-t_{j-1}$ in epoch
  $j$, we have $\sum_{C} y_{C}^{(j)} \geq \sum_{C \in\C_j}y_{C}^{(j)} = (t_j-t_{j-1}) |\C_{act}|$.   Thus, it suffices to show that $\sum_{C
    \in \C_{act}} \deg_H(C) \leq 2 |\C_{act}|$.

First we
  must make some simple observations about $H$.  Since $F_2$ is a subset of the edges in $F_1$, and each cluster represents a
  disjoint induced connected subtree of $F_1$, the contraction to $H$ introduces  no cycles.  Thus, $H$ is a forest.
  All the leaves of $H$ must
  be alive, because otherwise the corresponding cluster $C$ would be
  in $\B$ and  hence would have been pruned away.

  With this information about $H$, it is easy to bound $\sum_{C \in
    \C_{act}} \deg_H(C)$.
  The total degree in $H$ is at most $2(|\C_{act}|+|\C_{dead}|)$.
  Noticing that the degree of dead clusters is at least two,
  we get $\sum_{C\in\C_{act}}\deg_H(C) \leq 2(|\C_{act}|+|\C_{dead}|) - 2|\C_{dead}| = 2|\C_{act}|$ as desired.
\end{proof}
\fi

The following lemma gives a sufficient condition for two vertices that end up in the same component of $F_2$.
This is a corollary of our pruning rule which has a major difference from other pruning rules.
Unlike the previous work,
we do not prune the entire subgraph; rather, we only remove some edges, increasing the number of connected components.
\begin{lemma}\label{lem:conn}
 Two vertices $u$ and $v$ of $V$ are connected via $F_2$
if there exist sets $S, S'$ both containing $u, v$ such that
$y_{S,v} > 0$ and $y_{S',u}>0$.
\end{lemma}
\begin{proof}
 The growth stage connects $u$ and $v$ since $y_{S,v} > 0$ and $u,v\in S$.
 Consider the path $p$ connecting $u$ and $v$ in $F_1$.
 All the vertices of $p$ are in $S$ and $S'$.
 For the sake of reaching  a contradiction, suppose some edges of $p$ are pruned.
 Let $e$ be the first edge being pruned on the path $p$.
 Thus, there must be a cluster $C\in\B$ cutting $e$; furthermore, $\delta(C)\cap p=\{e\}$, since $e$ is the first edge pruned from $p$.
 The laminarity of the clusters $\eS$ gives $C\subset S,S'$,
since $C$ contains exactly one endpoint of $e$.
 If $C$ contains both or no endpoints of $p$, it cannot cut $p$ at only one edge.
 Thus, $C$ containts exactly one endpoint of $p$, say $v$.
 We then have $\sum_{C'\subseteq C}y_{C',v} = \phi_v$,  because $C$ is tight.
 However, as $C$ is a \emph{proper} subset of $S$,
 this contradicts with $y_{S,v}>0$, proving the supposition is
 false.  
 The case $C$ contains $u$ is symmetric.
\end{proof}

Consider a pair $(v,S)$ with $y_{S,v} > 0$.
If subgraph $G'$ of $G$ has an edge that goes through the cut $(S,\bar{S})$,
at least a portion of length $y_{S,v}$ of $G'$ is colored with the color $v$ due to the set $S$.
Thus, if  $G'$ cuts all the sets $S$ for which $y_{S,v} > 0$,
we can charge part of the length of $G'$ to the potential of $v$.
Later in Lemma~\ref{lem:break:sep}, we are going to use potentials as a lower bound on the optimal solution.
More formally, we say a graph $G'(V,E')$ \emph{exhausts a color} $u$ if and only if $E'\cap\delta(S)\neq\emptyset$ for any $S: y_{S,u}>0$.
The proof of the following corollary is omitted here, however, it is implicit in the proof of Lemma~\ref{lem:break:sep} below.
We do not use this corollary explicity.
Nevertheless, it gives insight into the analysis below.
\begin{corollary}
If a subgraph $H$ of $G$ connects two vertices $u_1,u_2$ from different components of $F_2$ (which are contracted versions of the components in the initial $2$-approximate solution), then $H$ exhausts the color corresponding to at least one of $u_1$ and $u_2$.
\end{corollary}

We can relate the cost of a subgraph to the potential value of
the colors it exhausts.
\begin{lemma}\label{lem:color}
 Let $L$ be the set of colors exhausted by subgraph $G'$ of $G$.
 The cost of $G'(V,E')$ is at least $\sum_{v\in L}\phi_v$.
\end{lemma}
This is quite intuitive.
Recall that the $y$ variables color the edges of the graph.
Consider a segment on edges corresponding to cluster $S$ with color $v$.
At least one edge of $G'$ \emph{passes through} the cut $(S,\bar{S})$.
Thus, a portion of the cost of $G'$ can be charged to $y_{S,v}$.
Hence, the total cost of the graph $G'$ is at least as large as the total amount of colors paid for by $L$.
\iffull We now provide a formal proof.
\begin{proof}
 The cost of $G'(V,E)$ is
\begin{align*}
 \sum_{e\in E'} c_e
  &\geq  \sum_{e\in E'}\sum_{S: e\in\delta(S)} y_S &\text{by \eqref{eqn:lp:1}}\\
  &=     \sum_S |E'\cap\delta(S)|y_S \\
  &\geq  \sum_{S: E'\cap\delta(S)\neq\emptyset} y_S \\
  &=     \sum_{S: E'\cap\delta(S)\neq\emptyset}\sum_{v\in S} y_{S,v} \\
  &=     \sum_{v}\sum_{S\ni v: E'\cap\delta(S)\neq\emptyset} y_{S,v} \\
  &\geq  \sum_{v\in L}\sum_{S\ni v: E'\cap\delta(S)\neq\emptyset} y_{S,v} \\
  &=     \sum_{v\in L}\sum_{S\ni v} y_{S,v}, \\\intertext{because $y_{S,v}=0$ if $v\in L$ and $E'\cap\delta(S)=\emptyset$,}
  &=     \sum_{v\in L}\phi_v  &\text{by a tight version of \eqref{eqn:lp:2}}. &\qedhere
\end{align*}
\end{proof}
\fi

Recall that the trees $T^\ast_i$ of the $2$-approximate solution $F^\ast$ 
are contracted in $F_2$. 
 Construct $F$ from $F_2$ by uncontracting all these trees. 
 Let $F$ consist of tree components $T_i$.
 It is not difficult to verify that $F$ is indeed a forest, but
 we do not need this condition since we can always remove cycles to
 find a forest.
 Define $\DD_i := \{ (s,t)\in\DD: s,t\in V(T_i)\}$,
 and let $\OPT_i$ denote (the cost of) the optimal Steiner forest satisfying the demands $\DD_i$.

\begin{lemma}\label{lem:break:sep}
 $\sum_i\OPT_i \leq (1+\eps)\OPT$.
\end{lemma}

\begin{proof}
 If each tree component of $\OPT$ only satisfies demands from a single $\DD_i$, we are done.
 Otherwise, we build a solution $\OPT'$ consisting of
 forests $\OPT'_i$ for each $\DD_i$,
 such that $\sum_i \OPT'_i \leq (1+\eps)\OPT$.
 Then, $\OPT_i \leq \OPT'_i$ finishes the argument.
 Instead of explicitly identifying the forests $\OPT'_i$, we construct an \emph{inexpensive} set of trees $\T$ where each $T\in\T$ is associated with a subset $D_T\subseteq\DD$ of demands, such that
 \begin{itemize}
  \item $T\in\T$ connects all the demands in $D_T$,
  \item all the demands are satisfied, namely, $\cup_{T\in\T}D_T = \DD$, and
  \item for each $T\in\T$, there exists a $\DD_i$ where $D_T\subseteq\DD_i$.
 \end{itemize}
 Then, each $\OPT'_i$ above is merely a collection of several trees in $\T$.
 We now describe how $\T$ is constructed.
 Start with the optimal solution $\OPT$.
 Contract the trees $T^\ast_i$ of the $2$-approximate solution $F^\ast$ to build a graph $\hat\OPT$.
 Initially, the set $D_{\hat T}$ associated with each tree component $\hat T$ of $\hat\OPT$ specifies the set of demands connected via $\hat T$.
 If the color $u$ corresponding to a tree component $T^\ast_i$ is exhausted by $\hat\OPT$, add $T^\ast_i$ to $\T$, and let $D_{T^\ast_i}$ be the set of demands satisfied via $T^\ast_i$.
 Then, for any tree component $\hat T_i$ of $\hat\OPT$,
  remove $\cup_{T\in\T}D_T$ from $D_{\hat T_i}$.
 Finally, construct the uncontracted version of all the trees  $\hat T_i$, and add them to $\T$.

 The first condition above clearly holds by the construction.
 For the second condition, note that at the beginning all the demands are satisfied by the trees $\hat T_i$, and we only remove demands from $D_{\hat T_i}$ when they are satisfied elsewhere.
 The last condition is proved by contradiction.
 Suppose a tree $T\in\T$ satisfies demands from two different groups $\DD_1$ and $\DD_2$.
 Clearly, $T$ cannot be one of the contracted trees $T^\ast_i$, since those trees only serve a subset of one $\DD_i$.
 Hence, $T$ must be a tree component of $\OPT$.
 Let $u,v\in V$ be two vertices from different components of $F_2$ that are connected in  $T$.
 Since $u$ and $v$ are not connected in $F_2$, Lemma~\ref{lem:conn} ensures that for at least one of these vertices, say $u$, we have $y_{S,u} = 0 \;\forall S\ni u,v$.
 Thus, $y_{S,u} > 0$ implies $v\not\in S$.
 As $T$ connects $u$ and $v$, this means $T$ exhausts the color $u$.
 We reach a contradiction, because the demands corresponding to $u$ are still in $D_T$.

 Finally, we show that the cost of $\T$ is small.
 Recall that $\T$ consists of two kinds of trees:
  tree components of $\OPT$ as well as trees corresponding to exhausted colors of $\OPT$.
 We claim the cost of the latter is at most $\eps\OPT$.
 Let $L$ be the set of exhausted colors of $\OPT$.
 The cost of the corresponding trees is $\eps\sum_{v\in L}\phi_v$ by definition of $\phi_v$.
 Lemma~\ref{lem:color}, however, gives $\OPT\geq\sum_{v\in L}\phi_v$.
 Therefore, the total cost of $\T$ is at most $(1+\eps)\OPT$.
\end{proof}

Now, we are ready to prove the main theorem of this section.
\begin{proof}[\proofname\ of Theorem~\ref{thm:break}]
 The first condition of the lemma follows directly from our construction:  we start with a solution, and never disconnect one of the tree components in the process.
 The construction immediately implies the second condition.
 By Lemma~\ref{lem:break:cost}, the cost of $F_2$ is at most $2\sum_{v\in V}\phi_v \leq \frac{4}{\eps}\OPT$. Thus, $F$ costs no more than $(4/\eps+2)\OPT$, giving the third condition.
 Finally, Lemma~\ref{lem:break:sep} establishes the last condition.
\end{proof}


\iffull
\section{Constructing the spanner}\label{sec:solve}\label{sec:spanner}
The goal of this section is to prove Theorem~\ref{thm:spanner}.
Recall that we are given a graph $G_{in}(V_{in}, E_{in})$, and a set
of demands $\DD$. From Theorem~\ref{thm:break}, we obtain a set of
trees $\{T_1, \dots, T_k\}$, associated with a partition of demands
$\{\DD_1, \dots, \DD_k\}$: tree $T_i$ connects all the demands
$\DD_i$, and the total cost of trees $T_i$ is in $O(\OPT)$.
The construction goes along the same lines as that of Borradaile et
al.~\cite{BKM07:planar}, yet there are certain differences,
especially in the analysis. The construction is done in four steps.
We separately build a graph $H_i$ for each $T_i$, and finally let
$H$ be the union of all graphs $H_i$.

\paragraph{Step 1: cutting the graph open}
Construct the graph $G_i^1$ as follows.
Start from $G_{in}$ and duplicate each edge in $T_i$.
Let $\mathcal{E}_i$ be the Eulerian tour of the duplicated $T_i$.
Introduce multiple copies of vertices in $T_i$ if necessary, so as to transform $\mathcal{E}_i$ into a simple cycle enclosing no vertices.
We next change the embedding so that this cycle is the infinite face of the graph.  It is clear that the length of the boundary of $G_i^1$ is twice the length of $T_i$.

\paragraph{Step 2: building panels}
Borradaile et al. pick a subgraph $G_i^2$ of $G_i^1$ which contains the boundary of $G_i^1$, such that its length is at most $f_1(\eps)\OPT_{\DD_i}(G_{in})$, and
partitions $G_i^1$ into \emph{panels}.
The panels are the finite faces of $G_i^2$.
We describe the details of this step below.

We first decompose $G_i^1$ into strips. Let $Q$ be
the boundary of $G_i^1$. Let $Q[x, y]$ denote the unique
nonempty counterclockwise $x$-to-$y$ subpath of $Q$.
We define $Q[x,x]:=Q$.
 We use a recursive algorithm:
 Find vertices $x, y$ on $Q$ such that\footnote{This is always possible, since $x=y$ satisfies these conditions.}
\begin{itemize}
\item $(1 + \eps)\dist_{G_i^1}(x, y) < \ell(Q[x, y])$, and
\item $(1 + \eps)\dist_{G_i^1}(x', y') \geq \ell(Q[x', y'])$ for every $x', y'$ in $Q[x, y]$ such that $x' \neq x$ or $y' \neq y$.
\end{itemize}
Let $B$ be a shortest path from $x$ to $y$ in $G_i^1$.
Then, the subgraph enclosed by $Q[x,y]\cup B$ is called a \emph{strip}.
We call the path $B$ the \emph{blue} boundary of the strip,
whereas the path $Q[x,y]$ is called the \emph{red} boundary of the strip,
and is denoted by $R$.
The blue and red boundaries are thought of as the upper and lower boundaries, respectively.
Then, both boundaries are oriented from left to right for the following discussion.

Recursively decompose the subgraph $G_i^1$ enclosed by $B\cup Q-Q[x,y]$ into strips (if it is nontrivial).
The next lemma is easy to prove.

\begin{lemma}[Inequality (10), Klein~\cite{klein:spanner}]\label{lem:strips}
The total length of all the boundary edges of all the strips is at
most $(\eps^{-1} + 1)$ times the length of $Q$.
\end{lemma}

After decomposing  $G_i^1$ into strips,
we find short paths---called \emph{columns} henceforth---crossing each strip. Consider a strip, and let $R$ and $B$ be its red and blue boundaries.

Vertices $r_0, r_1, \dots$ are inductively selected as follows.
Let $r_0$ be the left endpoint common to $R$ and $B$. For
$j = 1, 2, \dots$, find the vertex $r_j$ on $R$ such that
\begin{itemize}
\item $(1 + \eps)\dist_{G_i^1} (r_j, B) < \dist_R (r_j, r_{j-1})+ \dist_{G_i^1} (r_{j-1}, B)$, and
\item $(1 + \eps)\dist_{G_i^1} (x, B) \geq \dist_R (x, r_{j-1}) + \dist_{G_i^1} (r_{j-1}, B)$
for every $x$ in $R[r_{j-1}, r_j)$.
\end{itemize}
For $j = 0, 1, 2, \dots$, column $C_j$ is defined to be a
shortest path from $r_j$ to $B$. Note that $C_0$ is a path
with no edges since $r_0$ belongs to $B$. We also include as
a column the no-edge path starting and ending at the
rightmost vertex common to $R$ and $B$.
The total length of the columns can be bounded as follows.

\begin{lemma}[Lemma 5.2 of Klein~\cite{klein:spanner}]\label{lem:columns}
 The sum of the lengths of the columns of a strip with red boundary $R$ is at most $\eps^{-1}\ell(R)$.
\end{lemma}

Roughly speaking, for an integer parameter $\alpha=\alpha(\eps)$ whose precise value will be fixed later,
we select every $\alpha^{\rm th}$ column of each strip to finish the construction of $G_i^2$.
In particular, let $C_0, C_1, \dots, C_s$ be the set of columns of a strip.
For each $i = 0, 1, \dots, \alpha$, define the subset of columns $\C_i:=\{C_j: j=i\pmod\alpha\}$.
We select a group $\C_i$ whose total length is the least.
The selected columns are called \emph{supercolumns}.
Clearly, Lemma~\ref{lem:columns} implies the length of the supercolumns is at most $\eps^{-1}/\alpha \cdot \ell(R)$.

The supercolumns partition each strip into \emph{panels}.
Each panel is bounded above by a blue boundary, and below by a red boundary.
It is also bounded to the left and right by columns---recall that each strip is assumed to have trivial columns of length zero at its extreme left and right.
The value $\alpha$ will be sufficiently large such that we may add all the supercolumns to the optimal solution without significant increase in length.

\paragraph{Step 3: selecting portals}
They next select about $h=h(\eps)$ vertices on the red and blue boundaries of each panel as its \emph{portals}.  Denote the set of portals by $\Pi$. These are picked such that they are almost equidistant.

Consider a panel $P$, enclosed by $C_1, B, C_2, R$.
Define $\tau:=\frac{\ell(R)+\ell(B)}{h}$.
The left-most and right-most vertices of $R$ and $B$ are designated portals.
In addition, we designate a minimal set of vertices of $R$ and $B$ as portals, such that the distance between two consecutive portals on $R$ (or $B$) is at most $\tau$.
Ignoring the right-most (and left-most) portals, the minimality property guarantees the odd-numbered  portals (even-numbered portals, respectively) have distance larger than $\tau$.
Therefore, there are at most $2h+4=O(h)$ portals on each panel

\paragraph{Step 4: adding Steiner trees}
For each panel $P$, and any selection of its portals
$\Pi'\subseteq\Pi$, we add the optimal Steiner tree spanning $\Pi'$
and only using the boundary or the inside of $P$.  This is done
assuming that the left and right borders of the panel have length
zero.  This can be done in time polynomial in $h(\eps)$ using the
algorithm of Erickson et al.~\cite{Erickson}, since all these
terminals lie on the boundary face of a planar graph.  
Notice that for
a fixed $\epsilon$, there are at most a constant number of such
Steiner trees, and the length of each is at most the length of the
boundary of the panel.

The spanner $H_i$ consists of $G_i^2$, the Steiner trees built in the previous step, and edges of length zero connecting copies of a vertex on the infinite face.  It is easy to observe the latter edges could be added in such a way as to respect the planar embedding.
The following lemma is the main piece in proving the shortness property.
\begin{lemma}\label{lem:shortness}
 Length of $H_i$ is at most $f(\eps)\ell(T_i)$ for a universal certain function $f(\eps)$.
\end{lemma}
\begin{proof}
 By Lemma~\ref{lem:strips}, the sum of the lengths of the red boundaries of all strip boundaries is at most $(\eps^{-1}+1)\ell(Q)$.
 We can then apply Lemma~\ref{lem:columns} to bound the total length of $G_i^2$ by $(\eps^{-1}+1)^2\ell(Q)$.
 Since the length of each Steiner tree added in the last step is at most the length of the boundary of the panel,
 the total length of Steiner trees added is at most $2^{O(h)}$ times the sum of the length of all panels.
 The sum of the lengths of the panels is twice the length of $G_i^2$.
 Thus, the length of $H_i$ is at most $(2^{O(h)}+1)\ell(G_i^2)$.
 We finally notice that  the length of the boundary $Q$ of $G_i^1$ is $2\ell(T_i)$, to conclude with the desired bound.
\end{proof}

The spanning property is more involved to prove.
We first mention  a lemma proved in \cite{klein:wads} 
that states a few portal vertices are sufficient on each panel.
\begin{lemma}[Theorem 4 of \cite{klein:wads}]\label{lem:struct}
 Let $P$ be a panel enclosed by $C_1, B, C_2, R$.
 Further, let $F$ be a forest strictly enclosed by the boundary of $P$.
 There is a forest $\tilde F$ of $P$ with the following properties:
\begin{enumerate}
 \item If two vertices on $R\cup B$ are connected in $F$, then they are also connected in $\tilde F \cup C_1 \cup C_2$.
 \item The number of vertices in $R\cup B$ that are endpoints of edges in $\tilde F - (R\cup B)$ is at most $a(\eps)$.
 \item $\ell(\tilde F) \leq (1+c\eps)\ell(F)$.
\end{enumerate}
Here $c$ is an absolute constant and $a(\eps)=o(\eps^{-5.5})$. 
\end{lemma}

We are now at a position to prove the spanning property of $H$.
Recall that $H$ is formed by the union of the graphs $H_i$ constructed above.
\begin{lemma}\label{lem:spanning}
 $\OPT_{\DD}(H) \leq (1+c'\eps)\OPT_{\DD}(G_{in})$.
\end{lemma}
\begin{proof}
 Take the optimal solution $\OPT$.
 Find forests $\OPT_i$ satisfying demands $\DD_i$, i.e., $\ell(\OPT_i) = \OPT_{\DD_i}(G_{in})$.
 We can apply Theorem~\ref{thm:break} to get $\sum_i\ell(\OPT_i)\leq(1+\eps)\OPT$.

 Consider one $\OPT_i$ that serves the respective set of demands $\DD_i$.
 Add the set of all supercolumns of $H_i$ to $\OPT_i$ to get $\OPT_i^1$.
 Recall the total length of these supercolumns is at most $\ell(G_i^2)/k$.
 Pick $\alpha:=\eps^2/2(\eps^{-1}+1)^2$ to bound the supercolumns cost by $\eps^2\ell(T_i)$.
 Next, use Lemma~\ref{lem:struct} to replace the intersection of $\OPT_i^1$ and each panel with another forest having the properties of the lemma.
 Let $\OPT_i^2$ be the new forest.
 The cost of the solution increases by no more than $c\eps\ell(\OPT_i)$.
 Furthermore, as a result, $\OPT_i^2$ crosses each panel at most $a(\eps)$ times.
 We claim,  provided that $h$ is sufficiently large compared to $a$, we can ensure that moving these intersection points to the portals introduces no more than an $\eps$ factor to the cost.

Consider a panel $P$, with boundaries $C_1, B, C_2, R$.
Connect each intersection point of the panel to its closest portal.
Each connection on a panel $P$ moves by at most $\tau(P)$. (In fact, it is at most $\tau(P)/2$.)
The total movement of each panel is at most $a\tau(P)=a(\ell(R)+\ell(B))/h$.
Hence, the total additional cost for all panels of $H_i$  is
no more than $a/h\cdot \ell(G_i^2) = a/h\cdot (\eps^{-1}+1)^2\ell(T_i)$.
We pick $h(\eps) := \eps^2 a(\eps) (\eps^{-1}+1)^2$.
Thus, the addition is at most $\eps^2\ell(T_i)$.

Finally, we replace the forests inside each panel $P$ by the Steiner trees provisioned in the last step of our spanner construction.
Take a panel $P$ with the set of portals $\Pi$.
Let $K_1, K_2, \dots$ be the connected components of $\OPT_i^2$ inside $P$.
Each intersection point is connected to a portal of $P$.
Replace each $K_j$ by the optimal Steiner tree connecting corresponding to this subset of portals.
This procedure does not increase the length and produces a graph $\OPT_i^3$.

Clearly, $\OPT_i^3$ satisfies all the demands in $\DD_i$.
Thus, the union of all forests $\OPT_i^3$, henceforth referred to as $\OPT^\ast$, gives a solution for the given Steiner forest instance.
It only remains to bound the cost of $\OPT^\ast$.
We know that
\begin{align}
\ell(\OPT^\ast) &\leq \sum_i\OPT_i^3, \\
\intertext{which may be strict due to common edges between different $\OPT_i^3$ forests.
Replacing the Steiner trees by the optimal Steiner trees between portals cannot increase the cost, so, the only cost increase comes from connections to portals.  Thus, we get
}
                &\leq \sum_i (\OPT_i^2 + \eps^2\ell(T_i)). \label{eqn:5}
\end{align}
From the above discussion,
\begin{align}
\ell(\OPT_i^2) &\leq (1+c\eps)\ell(\OPT_i^1), &\text{by Lemma~\ref{lem:struct}, and}\label{eqn:6}\\
\ell(\OPT_i^1) &\leq \ell(\OPT_i)+\eps^2\ell(T_i) &\text{due to supercolumns' length}. \label{eqn:7}
\end{align}
Hence, the cost of $\OPT^\ast$ is
\begin{align*}
\ell(\OPT^\ast)
   &\leq \sum_i \left[(1+c\eps)\OPT_i + \eps^2(2+c)\ell(T_i)\right] &\text{by \eqref{eqn:5}, \eqref{eqn:6} and \eqref{eqn:7}}\\
   &=    \sum_i \left[(1+c\eps)\OPT_i\right] + \sum_i\left[\eps^2(2+c)\ell(T_i)\right] \\
   &\leq (1+c\eps)(1+\eps)\OPT  +  \eps^2(2+c)\sum_i\ell(T_i) &\text{by Theorem~\ref{thm:break}} \\
   &\leq (1+c\eps)(1+\eps)\OPT  +  \eps^2(2+c)(4/\eps+2)\OPT &\text{by Theorem~\ref{thm:break}}\\
   &= \left[1 + (c+1)\eps+\eps^2 + (2+c)(4\eps + 2\eps^2) \right]\OPT \\
   &\leq    (1+c'\eps)\OPT,
\end{align*}
if we pick $c' \leq (c+1)+\eps+(2+c)(4+2\eps)$.
This can be achieved via $c' := 14 + 7c$ assuming $\eps \leq 1$.
\end{proof}

\paragraph{Runnig time}
In the above proof, $c'$ is a fixed constant independent of $\eps$.
It can be easily verified that $h$ and $\alpha$ have polynomial dependence on $\eps$.  Therefore, $f(\eps)$ has a singly exponential dependence on $\eps$.
Borradaile et al.\ show how to carry out the above steps in time $O(n\log n)$ for each $H_i$, with the contant having singly exponential dependence on $\eps$.
There are at most $n$ such subgraphs, hence, the running time of this step is $O(n^2\log n)$.
Furthermore, Algorithm~\ref{alg:break} can be implemented to run in time $O(n^2\log n)$.
We do not explicitly store all the $y_{S,v}$ variables; rather, we only keep track of those which are relevant in the algorithm.
The growth stage has $O(n)$ iterations, and as a result, the set $\eS$ will have size $O(n)$.
A simple implementation guarantees $O(n^2\log n)$ time for the growth step; see~\cite{GW95}.
The pruning stage can also be implemented with a similar time bound:
We maintain a priority queue for the clusters in $\B$---the key for each cluster $S$ is $|F\cap\delta(S)|$.
Notice that there are at most $O(n)$ iterations in the \textbf{while} loop.
In each iteration, we perform $O(n)$ updates to the clusters of the queue as a result of removing an edge.

\begin{proof}[Proof of Theorem~\ref{thm:spanner}]
 Immediate from Lemmas~\ref{lem:shortness} and \ref{lem:spanning} if
 we use $\eps/c$ instead of $\eps$ in the above construction.
\end{proof}

\fi
\section{The PTAS for planar Steiner forest}\label{sec:main-ptas}

Having proved the spanner result, we can present our main PTAS for
\prob{Steiner forest} on planar graphs in this section. We first
mention two main ingredients of the algorithm. We invoke the
following result due to Demaine, Hajiaghayi and Mohar~\cite{DHM07}.
\begin{theorem}[\cite{DHM07}]\label{thm:contraction-decomposition}
 For a fixed genus $g$, and any integer
$k \geq 2$, and for every graph $G$ of Euler genus at most $g$,
the edges of $G$ can be partitioned into $k$ sets such that
contracting any one of the sets results in a graph of
treewidth at most $O(g^2 k)$. Furthermore, such a partition
can be found in $O(g^{5/2} n^{3/2} \log n)$ time.
\end{theorem}
As a corollary, this holds for a planar graph (which has genus 1).

We can now prove the main theorem of this
paper. Algorithm~\ref{alg:main-ptas} (\algo{PSF-PTAS}) shows the steps
of the PTAS.
\begin{algorithm}
\caption{\algo{PSF-PTAS}\label{alg:main-ptas}}
\textbf{Input:} planar graph $G_{in}(V_{in}, E_{in})$, and set of demands $\DD$.\\
\textbf{Output:} Steiner forest $F$ satisfying $\DD$.
\begin{algorithmic}[1]
 \STATE Construct the Steiner forest spanner $H$.
 \STATE Let $k \gets 2f(\eps)/\bar\eps$.
 \STATE Let $\eps \gets \min(1,\bar\eps/6)$.
 \STATE Using Theorem~\ref{thm:contraction-decomposition}, partition the edges of $H$ into $E_1, \dots, E_k$.
 \STATE Let $i^\ast \gets \arg\min_i \ell(E_i)$.
 \STATE Find a $(1+\eps)$-approximate Steiner forest $F^\ast$ of $\DD$ in $H/E_{i^\ast}$ via Theorem~\ref{th:boundedtwptas}.
 \STATE Output $F^\ast\cup E_{i^\ast}$.
\end{algorithmic}
\end{algorithm}
\begin{proof}[Proof of Theorem~\ref{thm:main}]
Given are planar graph $G_{in}$, and the set of demand pairs $\DD$.
We build a Steiner forest spanner $H$ using Theorem~\ref{thm:spanner}.
For a suitable value of $k$ whose precise value will be fixed below,
we apply Theorem~\ref{thm:contraction-decomposition} to partition the
edges of $H$ into $E_1, E_2, \dots, E_k$.
Let $E_{i^\ast}$ be the set having the least total length.
The total length of edges in $E_{i^\ast}$ is at most $\ell(H)/k$.
Contracting $E_{i^\ast}$ produces a graph $H^\ast$ of treewidth $O(k)$.

Theorem~\ref{th:boundedtwptas} allows us to find a solution $\OPT^\ast$ corresponding to $H^\ast$.
Adding the edges $E_{i^\ast}$ clearly gives a solution for $H$ whose
value is at most $(1+\eps)\OPT_\DD(H) + \ell(H)/k$.
Letting $\eps = \min(1,\bar\eps/6)$ and $k = 2f(\eps)/\bar\eps$  guarantees that
the cost of this solution is
\begin{align*}
&\leq (1+\eps)^2\OPT_\DD(G_{in})  +\ell(H)/k &\text{by Theorem~\ref{thm:spanner}}\\
&\leq (1+\eps)^2\OPT_\DD(G_{in})  + \frac{\bar\eps}{2}\OPT_\DD(G_{in})&\text{by Theorem~\ref{thm:spanner} and choice of $k$}\\
&(1+\bar\eps)\OPT_\DD(G_{in}) &\text{by the choice of $\eps$}.&\qedhere
\end{align*}
\end{proof}

The running time of all the algorithm except for the bounded-treewidth PTAS is bounded by $O(n^2\log n)$.
The parameter $k$ above has a singly exponential dependence on $\eps$.
Yet, the running time of the current procedure for solving the
bounded-treewidth instances is not bounded by a low-degree polynomial;
rather, $k$ amd $\epsilon$ appear in the exponent in the polynomial.
Were we able to improve the running time of this procedure, we would obtain a PTAS that runs in time $O(n^2\log n)$.

\subsection{Extension to bounded-genus graphs}\label{sec:genus}
We can generalize this result to the case of bounded-genus graphs.
Theorem~\ref{thm:break} does not assume any special structure for the input graph.
The ideas of Section~\ref{sec:spanner} can be generalized to the case of bounded-genus graphs similarly to the work of Borradaile et al.~\cite{BDT09:genus}.  This process does not increase the Euler genus of the graph, since the resulting graph has a subset of original edges.
Theorem~\ref{thm:contraction-decomposition} works for such graphs as well, and hence, as in Theorem~\ref{thm:main}, we can reduce the problem to a bounded-treewidth graph on which we apply Theorem~\ref{th:boundedtwptas}.

\begin{theorem}
 For any constant $\eps > 0$, there is a polynomial-time $(1+\eps)$-approximation algorithm for the graphs of Euler genus at most $g$.
\end{theorem}

\ifappendix
\section{Missing proofs from
  Section~\ref{sec:ptas-bound-treew}}\label{sec:missing-btw}
\fi
\ifmain  
\section{PTAS for graphs of bounded treewidth}\label{sec:ptas-bound-treew}
\iffull
The purpose of this section is to prove
Theorem~\ref{th:boundedtwptas} by presenting a PTAS for
\prob{Steiner forest} on graphs of bounded treewidth.
\subsection{Groups}
\fi

We define a notion of group that will be crucial in the algorithm. A
group is defined by a set $S$ of center vertices, a set $X$ of
``interesting'' vertices, and a maximum distance $r$; the {\em group} $\Group_G(X,S,r)$
contains $S$ and those vertices of $X$ that are at distance at most
$r$ from some vertex in $S$.

\begin{lemma}\label{lem:treecluster}
Let $T$ be a Steiner tree of $X\subseteq V(G)$ with cost
$W$. For every  $\epsilon>0$, there is a set $S\subseteq X$ of
$O(1+1/\epsilon)$ vertices such that $X=\Group_G(X,S,\epsilon W)$.
\end{lemma}
\fi
\iffullorapp
\ifappendix \subsection{Proof of Lemma~\ref{lem:treecluster}}\fi
\begin{proof}
  Let us select vertices $s_1$, $s_2$, $\dots$ from $X$ as long as
  possible, with the requirement that the distance of $s_i$ is more
  than $\epsilon W$ from every $s_j$, $1\le j < i$. Suppose that $s_t$
  is the last vertex selected this way. We claim that $t\le
  1+2/\epsilon$. Consider a shortest closed tour in $G$ that visits
  the vertices $s_1$, $\dots$, $s_t$. As the distance between any two
  such vertices is more than $\epsilon W$, the total length of the tour
  is more than $t\epsilon W$ (assuming that $t> 1$). On the other
  hand, all these vertices are on the tree $T$ and it is well known
  that there is a closed tour that visits every vertex of the tree in
  such a way that every edge of the tree is traversed exactly twice
  and no other edge of the graph is used. Hence the shortest tour has
  length at most $2W$ and $t\le 2/\epsilon$ follows.
\end{proof}
The following consequence of the definition of group is easy to see:
\begin{proposition}\label{prop:joincluster}
If $S_1,S_2,X_1,X_2$ are subsets of vertices and $r_1,r_2$ are real
numbers, then \[\Group_G(X_1,S_1,r_1)\cup \Group_G(X_2,S_2,r_2)\subseteq
\Group_G(X_1\cup X_2,S_1\cup S_2,\max\{r_1, r_2\}).\]
\end{proposition}
\fi

\iffull
\subsection{Conforming solutions}
\fi

\ifmain
Let $\B=(B_i)_{i=1\dots n}$ be the bags of a rooted nice tree
decomposition of width $k$. Let $V_i$ be the set of vertices appearing in $B_i$ or
in a descendant of $B_i$. Let $A_i$ be the set of {\em active
  vertices} at bag $B_i$: those vertices $v\in V_i$ for which there is
a demaind $\{v,w\}\in \DD$ with $w\not \in V_i$. Let $G_i:=G[V_i]$. A
Steiner forest $F$ induces a partition $\pi_i(F)$ of $A_i$ for every
$i=1,\dots, n$: let two vertices of $A_i$ be in the same class of
$\pi_i(F)$ if and only if
they are in the same component of $F$. Note that if $F$ is restricted
to $G_i$, then a component of $F$ can be split into up to $k+1$
components, thus $\pi_i(F)$ is a coarser partition then the partition
defined by the components of the restriction of $F$ to $G_i$.

Let $\Pi=(\Pi_i)_{i=1\dots n}$ be a collection such that $\Pi_i$ is a
set of partitions of $A_i$.  If $\pi_i(F)\in \Pi_i$ for every bag
$B_i$, then we say that $F$ {\em conforms} to $\Pi$.
\iffull The aim of this subsection is to give a polynomial-time
algorithm for bounded treewidth graphs that finds a minimum cost
solution conforming to a given $\Pi$ (for fixed $k$, the running time
is polynomial in the size of the graph and the size of the collection
$\Pi$ on graph with treewidth at most $k$). In
Section~\ref{sec:constr-part}, we construct a polynomial-size
collection $\Pi$ such that there is a $(1+\epsilon)$-approximate
solution that conforms to $\Pi$. Putting together these two results,
we get a PTAS for the Steiner forest problem on bounded treewidth
graphs. 
 \fi
\begin{lemma}\label{lem:findconforming}
For every fixed $k$, there is a polynomial time algorithm that, given
a graph $G$ with treewidth at most $k$ and a collection $\Pi$, finds
the minimum cost Steiner forest conforming to $\Pi$.
\end{lemma}
\fi
\iffullorapp
\ifappendix \subsection{Proof of Lemma~\ref{lem:findconforming}}\fi
The proof of Lemma~\ref{lem:findconforming} follows the standard
dynamic programming approach, but it is not completely trivial. First,
we use a technical trick that makes the presentation of the dynamic
programming algorithm simpler. We can assume that every terminal
vertex $v$ has degree 1: otherwise, moving the terminal to a new
degree 1 vertex $v'$ attached to $v$ with an edge $vv'$ having length
0 does not change the problem and does not increase treewidth. Thus by
Lemma~\ref{lem:nicer}, it can be assumed that we have a nice tree
decomposition $(T,\B)$ of width at most $k$ where no terminal vertex
is introduced and the join nodes contain no terminal vertices.
For the rest of the section, we fix such a tree decomposition and
notation $V_i$, $A_i$, etc.~refer to this fixed decomposition.

Let us introduce terminology and notation concerning partitions. A
partition $\alpha$ of a set $S$ can be considered as an equivalence
relation on $S$. Hence we use notation $(x,y)\in \alpha$ to say that
$x$ and $y$ are in the same class of $\alpha$. We denote by
$x^{\alpha}$ the class of $\alpha$ that contains element $x$.

If $F$ is a subgraph of $G$ and $S\subseteq V(G)$, then $F$ {\em
  induces} a partition $\alpha$ of $S$: $(x,y)\in \alpha$ if and only
if $x$ and $y$ are in the same component of $F$ (and every $x\in
S\setminus V(F)$ forms its own class). We say that partition $\alpha$
is {\em finer} than partition $\beta$ if $(x,y)\in \alpha$ implies
$(x,y)\in \beta$; in this case, $\beta$ is {\em coarser} than
$\alpha$. We denote by $\alpha_1\vee \alpha_2$ the unique finest
partition $\alpha$ coarser than both $\alpha_1$ and $\alpha_2$. This
definition is very useful in the following situation. Let $F_1$, $F_2$
be subgraphs of $G$, and suppose that $F_1$ and $F_2$ induces
partitions $\alpha_1$ and $\alpha_2$ of a set $S\subseteq V(G)$,
respectively. If $F_1$ and $F_2$ intersect only in $S$, then the
partition induced by the union of $F_1$ and $F_2$ is exactly
$\alpha_1\vee \alpha_2$. Let $\beta_i$ be a partition of $B_i$ for
some $i\in I$ and let $F_i$ be a subgraph of $G[V_i]$. Then we denote
by $F_i+\beta_i$ the graph obtained from $F_i$ by adding a new edge
$xy$ for every $(x,y)\in \beta_i$.

Following the usual way of designing algorithms for bounded-treewidth
graphs, we define several subproblems for each node $i\in I$. A
subproblem at node $i$ corresponds to finding a subgraph $F_i$ in
$G_i$ satisfying certain properties: $F_i$ is supposed to be the
restriction of a Steiner forest $F$ to
$V_i$. The properties defining a subproblem prescribe how
$F_i$ should look like form the ``outside world'' (i.e., from the part
of $G$ outside $V_i$) and they contain all the
information necessary for deciding whether $F_i$ can be extended, by edges
outside $V_i$, to a conforming solution. Let us discuss briefly and informally
what information these prescriptions should contain. Clearly, the edges of $F_i$ in
$B_i$ and the way $F_i$ connects the vertices of $B_i$ (i.e., the
partition $\alpha$ of $B_i$ induced by $F_i$) is part of this
information. Furthermore, the way $F_i$ partitions $A_i$ should also
be part of this information. However, there is a subtle detail that
makes the description of our algorithm significantly more technical.
The definition of $\pi_i(F)=\pi$ means that the components of $F$
partition $A_i$ in a certain way. But the restriction $F_i$ of $F$ to
$V_i$ might induce a finer partition of $A_i$ than $\pi$: it is possible that two
components of $F_i$ are in the same component of $F$. This means that we cannot
require that the partition of $A_i$ induced by $F_i$ belongs to
$\Pi$. We avoid this problem by ``imagining'' the partition $\beta$ of
$B_i$ induced by the full solution $F$, and require that $F_i$
partition $A_i$ according to $\pi$ if each class of $\beta$ becomes
connected somehow. In other words, instead of requiring that $F_i$ itself
partitions $A_i$ in a certain way, we require that $F_i+\beta$
induces a certain partition.

Formally, each subproblem $P$ is defined by a
tuple $(i,H,\pi, \alpha,\beta,\mu)$, where
\begin{enumerate}
\item[(S1)] $i\in I$ is a node of $T$,
\item[(S2)] $H$ is a spanning subgraph of $G[B_i]$ (i.e., contains all
  vertices of $G[B_i]$),
\item[(S3)] $\pi\in \Pi_i$ is a partition of $A_i$,
\item[(S4)] $\alpha$, $\beta$ are partitions of $B_i$ and $\beta$ is
  coarser than $\alpha$,
\item[(S5)] $\mu$ is an injective mapping from the classes of $\pi_i$ to the classes of
  $\beta$.
\end{enumerate}
The solution $c(i,H,\pi, \alpha,\beta,\mu)$ of a subproblem $P$ is the
minimum cost of a subgraph $F_i$ of $G[V_i]$ satisfying all of the
following requirements:
\begin{enumerate}
\item[(C1)] $F_i[B_i]=H$ (which implies $B_i\subseteq V(F_i)$).
\item[(C2)] $\alpha$ is the partition of $B_i$ induced by $F_i$.
\item[(C3)] The partition of $A_{i}$
  induced by
  $F_i+\beta$ is $\pi$.
\item[(C4)] For every descendant $d$ of $i$, the partition of $A_{d}$
  induced by
  $F_i+\beta$ belongs to $\Pi_{d}$.
\item[(C5)] If there is a terminal pair $(x_1,x_2)$ with $x_1,x_2\in V_i$, then
 they are connected in $F_i+\beta$.
\item[(C6)]
Every $x\in A_i$ is in the component of $F_i+\beta$ containing $\mu(x^{\pi})$.
\end{enumerate}

We solve these subproblems by bottom-up dynamic programming. Let us
discuss how to solve a subproblem depending on the type of $i$.

\textbf{Leaf nodes $i$.} If $i$ is a leaf node, then value of the
solution is trivially 0.

\textbf{Join node $i$ having children $i_1$, $i_2$.} Note that
$A_{i_1}$ and $A_{i_2}$ are disjoint: the vertices of a join node are
not terminal vertices. The set $A_i$ is a subset of $A_{i_1}\cup
A_{i_2}$ and it may be a proper subset: if there is a pair $(x,y)$ with
$x\in A_{i_1}$, $y\in A_{i_2}$, then $x$ or $y$ might not be in $A_i$.

We show that the value of the
subproblem is
\begin{equation}
c(i,H,\pi, \alpha,\beta,\mu)=\min_{(J1),(J2),(J3),(J4)}(c(i_1,H,\pi^1,
\alpha^1,\beta,\mu^1)+c(i_2,H,\pi^2,
\alpha^2,\beta,\mu^2)-\ell(H)), \label{eq:join}
\end{equation}
where the minimum is taken
over all tuples satisfying, for $\myell=1,2$, all of the following
conditions:
\begin{enumerate}
\item[(J1)] $\alpha^1\vee \alpha^2=\alpha$.
\item[(J2)] $\pi$ and $\pi^{\myell}$ are the same on $A_{i_\myell}\cap
  A_i$.
\item[(J3)] For every $v\in A_{i_\myell}\cap A_i$, $\mu(v^{\pi})=\mu^\myell(v^{\pi^\myell})$.
\item[(J4)] If there is a terminal pair $(x_1,x_2)$ with $x_1\in A_1$ and
  $x_2\in A_2$, then $\mu^1(x_1^{\pi^1})=\mu^2(x_2^{\pi^2})$.
\end{enumerate}

We will use the following observation repeatedly. Let $F$ be subgraph
of $G_i$ and let $F^\myell=F[V_{i_\myell}]$. Suppose that $F$ induces
partition $\alpha$ on $B_i$ and $\beta$ is coarser than $\alpha$.
Then two vertices of $V_{i_\myell}$ are connected in $F+\beta$ if and only
if they are connected in $F^\myell+\beta$. Indeed, $F^{3-\myell}$ does not
provide any additional connectivity compared to $F^\myell$: as $\beta$ is coarser than $\alpha$,
if two vertices of $B_i$ are connected in $F^{3-\myell}$, then they are
already adjacent in $F^\myell+\beta$.
\medskip

\dynproof{eq:join}{$\le$}

Let $P_1=(i_1,H,\pi^1,
\alpha^1,\beta,\mu^1)$ and $P_2=(i_2,H,\pi^2,
\alpha^2,\beta,\mu^2)$ be subproblems minimizing the right hand side
of \eqref{eq:join}, let $F^1$
and $F^2$ be optimum solutions of $P_1$ and $P_2$,
respectively.
Let $F$ be the union of subgraphs
$F^1$ and $F^2$. It is clear that the cost of $F$ is exactly the
right hand side of \eqref{eq:join}: the common edges of $F^1$ and
$F^2$ are exactly the edges of $H$. We show that $F$ is a solution of
$P$, i.e., $F$ satisfies
requirements (C1)--(C6).

\noindent (C1): Follows from $F^1[B_i]=F^2[B_i]=F[B_i]=H$.

\noindent (C2): Follows from (J1) and from the
fact that $F^1$ and $F^2$ intersect only in $B_i$.

\noindent (C3): First consider two vertices $x,y\in A_{i_\myell}\cap A_i$.
Vertices $x$ and $y$ are connected in $F+\beta$ if and only if they
are connected in $F^\myell+\beta$, which is true if and only if
$(x,y)\in \pi^\myell$ holds, which is equivalent to $(x,y)\in \pi$ by
(J2). Now suppose that $x\in A_{i_1}\cap A_i$ and $y\in A_{i_2}\cap A_i$. In this
case, $x$ and $y$ are connected in $F+\beta$ if and only if there is a
vertex of $B_i$ reachable from $x$ in $F^1+\beta$ and from $y$ in
$F^2+\beta$, or in other words,
$\mu^1(x^{\pi^1})=\mu^2(y^{\pi^2})$. By (J3), this is equivalent to
$\mu(x^\pi)=\mu(x^\pi)$, or $(x,y)\in \pi$ (as $\mu$ is injective).

\noindent (C4): If $d$ is a descendant of $i_\myell$, then the statement
follows using that (C4) holds for solution $F^\myell$ of $P^\myell$ and
the fact that for every descendant $d$ of
$i_\myell$, $F_i+\beta$ and $F+\beta$ induces the same partition of
$A_d$. For $d=i$, the statement follows from the previous paragraph,
i.e., from the fact that $F+\beta$ induces partition $\pi\in \Pi_i$ on $A_i$.

\noindent (C5): Consider a pair $(x_1,x_2)$. If $x_1,x_2\in V_{i_1}$ or $x_1,x_2\in V_{i_2}$,
then the statement follows from (C5) on $F^1$ or $F^2$. Suppose now
that $x_1\in V_{i_1}$ and $x_2\in V_{i_2}$; in this case, we have
$x_1\in A_{i_1}$ and $x_2\in A_{i_2}$.  By (C6) on $F^1$ and $F^2$,
$x_\myell$ is connected to
$\mu^\myell(x_\myell^{\pi^\myell})$ in $F^\myell+\beta$. By (J4), we have
$\mu^1(x_1^{\pi^1})=\mu^2(x_2^{\pi^2})$, hence $x_1$ and $x_2$ are
connected to the same class of $\beta$ in $F+\beta$.

\noindent (C6): Consider an $x\in A_i$ that is in $A_{i_\myell}$.  By
condition (C6) on $F^\myell$, we have that $x$ is connected in $F^\myell+\beta$
(and hence in $F+\beta$) to $\mu^\myell(v^{\pi^\myell})$, which
equals $\mu(v^{\pi})$ by (J3).
  \medskip

\dynproof{eq:join}{$\ge$}

Let $F$ be a solution of subproblem $(i,H,\pi, \alpha,\beta,\mu)$ and
let $F^\myell$ be the subgraph of $F$ induced by $V_{i_\myell}$.
To prove the inequality, we need to show three things. First, we have
to define two tuples $(i_1,H,\pi^1,
\alpha^1,\beta,\mu^1)$ and $(i_2,H,\pi^2,
\alpha^2,\beta,\mu^2)$ that are subproblems, i.e., they satisfy
(S1)--(S5). Second, we show that (J1)--(J4) hold for these
subproblems. Third, we show that $F^1$ and $F^2$ are solutions for
these subproblems (i.e., (C1)--(C6)), hence they can be used to give
an upper bound on the right hand side that matches the cost of $F$.

Let $\alpha^\myell$ be the partition of $B_i$ induced by the components
of $F^\myell$; as $F^1$ and $F^2$ intersect only in $B_i$, we have
$\alpha=\alpha^1\vee \alpha^2$, ensuring (J1). Since $\beta$ is
coarser than $\alpha$, it is coarser than both $\alpha^1$ and
$\alpha^2$. Let $\pi^\myell$ be the partition of $A_{i_\myell}$ defined by
$F+\beta$, we have $\pi_{\myell}\in \Pi_{i_\myell}$ by (C4) for $F$.
Furthermore, by (C3) for $F$, $\pi$ is the partition of $A_i$ induced
by $F+\beta$, hence it is clear that $\pi$ and $\pi^\myell$ are the same
on $A_{i_\myell}\cap A_i$, so (J2) holds. This also means that $F+\beta$
(or equivalently, $F^\myell+\beta$) connects a class of $\pi^\myell$ to
exactly one class of $\beta$; let $\mu^\myell$ be the corresponding
mapping from the classes of $\pi^\myell$ to $\beta$. Now (J4) is
immediate.  Furthermore, it is clear that the tuple
$(i_\myell,H,\pi^\myell, \alpha^\myell,\beta,\mu^\myell)$ satisfies
(S1)--(S5).

We show that $F^\myell$ is a solution of subproblem $(i_\myell,H,\pi^\myell,
\alpha^\myell,\beta,\mu^\myell)$. As the edges of $H$ are
shared by $F^1$ and $F^2$, it will follow that the right hand side of
\eqref{eq:join} is not greater than the left hand side.

\noindent (C1): Obvious from the definition of $F^1$ and $F^2$.

\noindent (C2): Follows from the way $\alpha^\myell$ is defined.

\noindent (C3): Follows from the definition of $\pi^\myell$, and from
the fact that $F+\beta$ and $F^\myell+\beta$ induces the same partition
on $A_{i_\myell}$.

\noindent (C4): Follows from (C4) on $F$ and from the fact that $F+\beta$
and $F^\myell+\beta$ induces the same partition on $A_{d}$.

\noindent (C5): Suppose that $x_1,x_2\in
V_{i_\myell}$. Then by (C5) for $F$, $x_1$ and $x_2$ are connected in
$F+\beta$, hence they are connected in $F_i+\beta$ as well.

\noindent (C6): Follows from the definition of $\mu^\myell$.

\textbf{Introduce node $i$ of vertex $v$.}
Let $j$ be the child of $i$. Since $v$ is not a terminal vertex, we
have $A_j=A_i$. Let $F'$ be a subgraph of $G[V_j]$ and let $F$ be
obtained from $F$ by adding vertex $v$ to $F'$ and making $v$ adjacent
to $S\subseteq B_j$. If $\alpha'$ is the partition of $B_j$ induced by
the components of $F'$, then we define the partition $\alpha'[v,S]$
of $B_i$ to be the partition obtained by joining all the classes of $\alpha'$
that intersect $S$ and adding $v$ to this new class (if $S=\emptyset$,
then $\{v\}$ is a class of $\alpha'[v,S]$). It is clear that
$\alpha'[S,v]$ is the partition of $B_i$ induced by $F$.

 We show that the value of a subproblem is given by
\begin{equation}
c(i,H,\pi, \alpha,\beta,\mu)=\min_{(I1),(I2),(I3)}c(j,H[B_j],\pi,
\alpha',\beta',\mu')+\sum_{xv\in E(H)}\ell(xv),\label{eq:intro}
\end{equation}
where the minimum is taken over all tuples satisfying all of the
following:
\begin{itemize}
\item[(I1)] $\alpha=\alpha'[v,S]$, where $S$ is the set of neighbors of $v$
  in $H$.
\item[(I2)] $\beta'$ is $\beta$ restricted to $B_j$.
\item[(I3)] For every $x\in A_i$, $\mu(x^\pi)$ is the class of $\beta$
  containing $\mu'(x^\pi)$.
\end{itemize}

\dynproof{eq:intro}{$\le$}

Let $F'$ be an optimum solution of subproblem $P'=(j,H[B_j],\pi,
\alpha',\beta',\mu')$. Let $F$ be the graph obtained from $F'$ by
adding to it the edges of $H$ incident to $v$; it is clear that the
cost of $F$ is exactly the right hand side of \eqref{eq:intro}. Let
us verify that (C1)--(C6) hold for $F$.

\noindent (C1): Immediate.

\noindent (C2): Holds because of (I1) and the way $\alpha'[v,S]$ was defined.

\noindent (C3)--(C5): Observe that $F+\beta$ connects two
vertices of $V_j$ if and only if $F'+\beta'$ does. Indeed, if a path
in $F+\beta$ connects two vertices via vertex $v$, then the two
neighbors $x,y$ of $v$ on the path are in the same class of $\beta$ as
$v$, hence (by (I2)) $x,y$ are in the same class of $\beta'$ as
well. In particular,
for every descendant $d$ of $i$, the components of $F+\beta$ and the
components of $F'+\beta$  give the same partition of $A_d$.

\noindent (C6): Follows from (C6) for $F'$ and from (I3).
\medskip

\dynproof{eq:intro}{$\ge$}

Let $F$ be a solution of subproblem $(i,H,\pi, \alpha,\beta,\mu)$ and
let $F'$ be the subgraph of $F$ induced by $V_{j}$.  We define a tuple
$(j,H[B_j],\pi, \alpha',\beta',\mu')$ that is a subproblem, show that
it satisfies (I1)--(I3), and that $F'$ is a solution of this subproblem.

Let $\alpha'$ be the partition of $A_i=A_j$ induced by $F'$ and let
$\beta'$ be the restriction of $\beta$ on $B_j$; these definitions
ensure that (I1) and (I2) hold. Let $\mu'(x^{\pi})$ defined to be the
class of $\beta'$ arising as the restriction of the class
$\mu(x^{\pi})$ to $B_j$; clearly, this ensures (I3). Note that this is
well defined, as it is not possible that $\mu(x^{\pi})$ is a class of
$\beta$ consisting of only $v$: by (C6) for $F$, this would mean that
$v$ is the only vertex of $B_i$ reachable from $x$ in $F$. Since $v$
is not a terminal vertex, $v\neq x$, thus if $v$ is reachable from
$x$, then at least one neighbor of $v$ has to be reachable from $x$ as
well.

Let us verify that (S1)--(S5) hold for the tuple $(j,H[B_j],\pi,
\alpha',\beta',\mu')$. (S1) and (S2) clearly holds. (S3) follows from
the fact that (C4) holds for $F$ and $A_i=A_j$. To see that (S4)
holds, observe that $(x,y)\in \alpha'$ implies $(x,y)\in \alpha$,
which implies $(x,y) \in \beta$, which implies $(x,y)\in \beta'$. (S5)
is clear from the definition of $\mu'$.

The difference between the cost of $F$ and the cost of $F'$ is
exactly $\sum_{xv\in E(H)}\ell(xv)$. Thus to show that the left of
\eqref{eq:intro} is at most the right of \eqref{eq:intro}, it is
sufficient to show that $F'$ is a solution of subproblem
$(j,H[B_j],\pi, \alpha',\beta',\mu')$.

\noindent (C1) and (C2): Obvious.

\noindent (C3)--(C5): As in the other direction, follows from the fact that $F'+\beta'$ induces the
same partition of $A_j=A_i$ as $F+\beta$.

\noindent (C6): By the definition of $\mu'$, it is clear that $x$ is
connected only to $\mu'(x^\pi)$ in $F+\beta$ and hence in $F'+\beta$.

\textbf{Forget node $i$ of vertex $v$.}
Let $j$ be the child of $i$. We have $V_i=V_j$ and hence $A_i=A_j$.
 We show that the value of a subproblem is given by
\begin{equation}
c(i,H,\pi, \alpha,\beta,\mu)=\min_{(F1),(F2),(F3),(F4)}c(j,H',\pi,
\alpha',\beta',\mu'),\label{eq:forget}
\end{equation}
where the minimum is taken over all tuples satisfying all of the
following:
\begin{itemize}
\item[(F1)] $H'[B_i]=H$.
\item[(F2)] $\alpha$ is the restriction of $\alpha'$ to $B_i$.
\item[(F3)] $\beta$ is the restriction of $\beta'$ to $B_i$ and
  $(x,v)\in \beta'$ if and only if $(x,v)\in \alpha'$.
\item[(F4)] For every $x\in A_i$, $\mu(x^\pi)$ is the (nonempty) set
  $\mu'(x^\pi)\setminus \{v\}$ (which implies that $\mu'(x^\pi)$ contains at least one
  vertex of $B_i$).
\end{itemize}

\dynproof{eq:forget}{$\le$}

Let $F$ be a solution of $(j,H',\pi,
\alpha',\beta',\mu')$. We show that $F$ is a solution of $(j,H,\pi,
\alpha,\beta,\mu)$ as well.

\noindent (C1): Clear because of (F1).

\noindent (C2): Clear because of (F2).

\noindent (C3)--(C5): We only need to observe
that $F+\beta$ and $F+\beta'$ have the same components: since by (F3), $(x,v)\in
\beta'$ implies $(x,v)\in \alpha'$, the neighbors of $v$ in $F+\beta$
are reachable from $v$ in $F$, thus $F+\beta'$ does not add any
further connectivity compared to $F+\beta$.

\noindent (C6): Observe that if $\mu'(x^\pi)$ are the vertices of
$B_j$ reachable from $x$ in $F+\beta'$, then $\mu'(x^\pi)\setminus
\{v\}$ are the vertices of $B_i$ reachable form $F+\beta$. We have
already seen that $F+\beta$ and $F+\beta'$ have the same components,
thus the nonempty set $\mu(x^{\pi})$ is indeed the subset of $B_i$
reachable from $x$ in $F+\beta$. Furthermore, by (F3), $\beta$ is the
restriction of $\beta'$ on $B_i$, thus if $\mu'(x^{\pi})$ is a class
of $\beta'$, then $\mu(x^{\pi})$ is a class of $\beta$.
\medskip

\dynproof{eq:forget}{$\ge$}

Let $F$ be a solution of $(j,H,\pi, \alpha,\beta,\mu)$.  We define a
tuple $(j,H',\pi, \alpha',\beta',\mu')$ that is a subproblem, we show
that (F1)--(F3) hold, and that $F$ is a solution of this subproblem.

Let us define $H'=F[B_j]$ and let $\alpha'$ be the
partition of $B_j$ induced by the components of $F$; these definitions
ensure that (F1) and (F2) hold.
We define $\beta'$ as the partition obtained by extending $\beta$ to
$B_j$ such that $v$ belongs to the class of $\beta$ that contains a
vertex $x\in B_i$ with $(x,v)\in \alpha'$ (it is clear that there is
at most one such class; if there is no such class, then we let $\{v\}$ be a
class of $\beta'$). It is clear that (F3) holds for this $\beta'$. Let us
note that $F+\beta$ and
$F+\beta'$ have the same connected components: if $(x,v)\in \beta'$,
then $x$ and $v$ are connected in $F$. Let $\mu'(x^{\pi})$ be the subset of $B_j$
reachable from $x$ in $F+\beta'$ (or equivalently, in $F+\beta$). It is
clear that $\mu(x^{\pi})=\mu'(x^{\pi})\setminus \{v'\}$ holds, hence
(F4) is satisfied.

Let us verify first that (S1)--(S5) hold for $(j,H',\pi,
\alpha',\beta',\mu')$. (S1) and (S2) clearly holds. (S3) follows from the
fact that (S3) holds for $(i,H,\pi, \alpha,\beta,\mu)$ and $A_i=A_j$.
To see that (S4) holds, observe that if $x,y\in B_i$, then $(x,y)\in
\alpha'$ implies $(x,y)\in \alpha$, which implies $(x,y) \in \beta$,
which implies $(x,y)\in \beta'$. Furthermore, if $(x,v)\in \alpha'$,
then by $(x,v)\in \beta'$ by the definition of $\beta$.  (S5) is clear
from the definition of $\mu'$.

We show that $F$ is a solution of $(j,H',\pi, \alpha',\beta',\mu')$.

\noindent (C1): Clear from the definition of $H'$.

\noindent (C2):  Clear from the
definition of $\alpha'$.

\noindent (C3)--(C5): Follow from the fact that $F+\beta$ and
$F+\beta'$ have the same connected components.

\noindent (C6): Follows from the
definition of $\mu'$.
\fi

\ifmain
\iffull
\subsection{Constructing the partitions}
\label{sec:constr-part}
\fi
Recall that the collection $\Pi$ contains a set of partitions $\Pi_i$
for each $i\in I$.
Each partition in
$\Pi_i$ is defined by a sequence $((S_1,r_1),\dots,(S_\myell,r_\myell))$
of at most $k+1$ pairs and a partition $\rho$ of $\{1,\dots, \myell\}$.
We will denote by $\rho(j)$ the class of the partition containing $j$.
The pair $(S_j,r_j)$ consists of a set $S_j$ of
$O((k+1)(1+1/\epsilon))$ vertices of $G_i$ and a
nonnegative integer $r_j$, which equals the distance between two
vertices of $G$. There are at most
$|V(G)|^{O((k+1)(1+1/\epsilon))}\cdot |V(G)|^2$ possible
pairs $(S_j,r_j)$ and hence at most
$|V(G)|^{O((k+1)^2(1+1/\epsilon))}$ different
sequences. The number of possible partitions $\rho$ is $O(k^k)$. Thus if we construct $\Pi_i$ by considering all possible
sequences constructed from every possible choice of $(S_j,r_j)$, the
size of $\Pi_i$ is polynomial in $|V(G)|$ for every fixed $k$ and $\epsilon$.

We construct the partition $\pi$ corresponding to a particular
sequence and $\rho$ the following way. Each pair $(S_j,r_j)$ can be
used to define a group $R_j=\Group_G(A_i,S_j,r_j)$ of $A_i$. Roughly
speaking, for each class $P$ of $\rho$, there is a corresponding class
of $\pi$ that contains the union of $R_j$ for every $j\in P$. However,
the actual definition is somewhat more complicated. We want $\pi$ to
be a partition, which means that the subsets of $A_i$ corresponding to
the different classes of $\rho$ should be disjoint. In order to ensure
disjointness, we define $R'_j:=R_j\setminus \bigcup_{j'=1}^{j-1} R_{j'}$. The
partition $\pi$ of $A_i$ is constructed as follows: for each class $P$
of $\rho$, we let $\bigcup_{j\in P}R_j$ be a class of $\pi$. Note that
these classes are disjoint by construction. If these classes fully
cover $A_i$, then we put the resulting partition $\pi$ into $\Pi_i$;
otherwise, the sequence does not define a partition.  This finishes
the construction of $\Pi_i$.

Before showing that there is a good approximate solution conforming to
the collection $\Pi$ defined above, we need a further definition.  For
two vertices $u$ and $v$, we denote by $u < v$ the fact that the
topmost bag containing $u$ is a proper descendant of the topmost bag
containing $v$. Note that each bag is the topmost bag of at most one
vertex in a nice tree decomposition (recall that we can assume that
the root bag contains only a single vertex). Thus if $u$ and $v$
appear in the same bag, then $u<v$ or $v<u$ holds, i.e., this relation
defines an ordering of the vertices in a bag. We can extend this
relation to connected subset of vertices: for two disjoint connected
sets $K_1$, $K_2$, $K_1<K_2$ means that $K_2$ has a vertex $v$ such
that $u<v$ for every vertex $u\in K_1$, in other words, $K_1<K_2$
means that the topmost bag where vertices from $K_1$ appear is a
proper descendant of the topmost bag where vertices from $K_2$ appear.
If there is a bag containing vertices from both $K_1$ and $K_2$, then
either $K_1<K_2$ or $K_2<K_1$ holds.  The reason for this is that the
bags containing vertices from $K_1\cup K_2$ form a connected subtree
of the tree decomposition, and if the topmost bag in this subtree
contains vertex $v\in K_1\cup K_2$, then $u<v$ for every other vertex
$u$ in $K_1\cup K_2$.

\begin{lemma}\label{lem:approx}
There is a $(1+k\epsilon)$-approximate solution conforming to $\Pi$.
\end{lemma}
\begin{proof}
  Let $F$ be a minimum cost Steiner forest. We describe a procedure
  that adds further edges to $F$ to transform it into a Steiner forest
  $F'$ that conforms to $\Pi$ and has cost at most
  $(1+k\epsilon)\ell(F)$. We need a delicate charging argument to show
  that the total increase of the cost is at most $k\epsilon \cdot
  \ell(F)$ during the procedure. In each step, we charge the increase of
  the cost to an ordered pair $(K_1,K_2)$ of components of $F$. We
  are charging only to pairs $(K_1,K_2)$ having the property that
  $K_1<K_2$ and there is a bag containing vertices from both $K_1$ and
  $K_2$. Observe that if $B_i$ is the topmost bag where vertices from
  $K_1$ appear, then these properties imply that a vertex of $K_2$
  appears in this bag as well. Otherwise, if every bag containing
  vertices of $K_2$ appears above $B_i$, then there is no bag
  containing vertices from both $K_1$ and $K_2$; if every bag
  containing vertices from $K_2$ appears below $B_i$, then $K_1<K_2$
  is not possible. Thus a component $K_1$ can be the first component
  of at most $k$ such pairs $(K_1,K_2)$: since the components are
  disjoint, the topmost bag containing vertices from $K_1$ can
  intersect at most $k$ other components.  We will charge a cost
  increase of at most $\epsilon \cdot \ell(K_1)$ on the pair $(K_1,K_2)$,
  thus the total increase is at most $k\epsilon \cdot \ell(F)$.  It is a
  crucial point of the proof that we charge on (pairs of) components
  of the original solution $F$, even after several modification steps,
  when the components of $F'$ can be larger than the original
  components of $F$. Actually, in the proof to follow, we will refer
  to three different types of components:
\begin{enumerate}
\item[(a)] Components of the current solution $F'$.
\item[(b)] Each component of $F'$ contains one or more components of
  $F$.
\item[(c)] If a component of $F$ is restricted to the subset $V_i$,
  then it can split into up to $k+1$ components.
\end{enumerate}
To emphasize the different meanings, and be clear as well, we use the terms a-component,
b-component, and c-component.

Initially, we set $F':=F$ and it will be always true that $F'$ is a
supergraph of $F$, thus $F'$ defines a partition of the b-components of
$F$. Suppose that there is a bag $B_i$ such that the partition
$\pi_i(F')$ of $A_i$ induced by $F'$ is not in $\Pi_i$. Let $K_1 < K_2
<\dots <K_\myell$ be the b-components of $F$ intersecting $B_i$, ordered
by the relation $<$. Some of these b-components might be in the same
a-component of $F'$; let $\rho$ be the partition of $\{1,\dots,\myell\}$
defined by $F'$ on these b-components.

Let $A_{i,j}$ be the subset of $A_i$ contained in $K_j$. The
intersection of b-component $K_j$ with $V_i$ gives rise to at most $k+1$
c-components, each of cost at most $\ell(K_j)$. Thus by
Lemma~\ref{lem:treecluster}\iffull and Prop.~\ref{prop:joincluster}\fi, there is
a set $S_j\subseteq V(K_j)$ of at most
$O((k+1)(1+1/\epsilon))$ vertices such that
$A_{i,j}= \Group_{G_i}(A_{i,j},S_j,r_j)$ for some $r_j \le \epsilon\cdot
\ell(K_j)$. If the sequence $(S_1,r_1)$, $\dots$, $(S_\myell,r_\myell)$ and the
partition $\rho$ give rise to the partition $\pi_i(F')$, then
$\pi_i(F')\in \Pi_i$. Otherwise, let us investigate the reason
why this sequence and $\rho$ do not define the partition $\pi_i(F')$.
There are two possible problems: either a set $R'_j$ arising in the
construction is too small (i.e., $A_{i,j}\not\subseteq R'_j$) or too
large (i.e., $R_j$ contains a vertex from some $A_{i,j^*}$ such that
$j$ and $j^*$ are not in the same class of $\rho$).  If none of these
problems arise, then it is clear that the constructed partition $\pi$
is indeed a partition and it is the same as the partition $\pi_i(F')$.

Let $j$ be the smallest integer such that $R_j$ is too large or too
small.  By the choice of $S_j$, $R_j=\Group_{G_i}(A_{i},S_j,r_j)$
contains $A_{i,j}$. If no $R_{j^*}$ with $j^*<j$ is too large, then no
such $R_{j^*}$ contains vertices from $A_{i,j}$, which means that
$R'_j:=R_j\setminus \bigcup_{j^*=1}^{j-1} R_{j^*}$ fully contains
$A_{i,j}$. Thus $R'_j$ is not too small. Suppose now that $R'_j$ is
too large: it contains a vertex of $A_{i,j^*}$ for some $j^*\not\in
\rho(j)$. It is not possible that $j^*<j$: by assumption, $R'_{j^*}$
is not too small, thus $R'_{j^*}$ fully contains $A_{i,j^*}$ and hence
$R'_j$ is disjoint from $A_{i,j^*}$. Thus we can assume that $j^*>j$.
The fact that $\Group_{G_i}(A_{i},S_j,r_j)$ intersects $A_{i,j^*}$
means that there is a vertex $u\in S_j$ and vertex $v\in A_{i,j^*}$
such that $d_{G_i}(u,v)\le \epsilon \cdot \ell(K_j)$.  Note that $u$ is a
vertex of b-component $K_j$ (as $S_j\subseteq A_{i,j}$) and $v$ is a
vertex of b-component $K_{j^*}$.  We modify $F'$ by adding a shortest
path that connects $u$ and $v$. Clearly, this increases the cost of
$F'$ by at most $\epsilon \cdot \ell(K_j)$, which we charge on the pair
$(K_j,K_{j^*})$. Note that $K_j$ and $K_{j^*}$ both intersect the bag
$B_i$ and $K_j<K_{j^*}$, as required in the beginning of the proof.
Furthermore, $K_j$ and $K_{j^*}$ are in the same a-component of $F'$
after the modification, but not before. Thus we charge at most once on
the pair $(K_j,K_{j'})$.

Since the modification always extends $F'$, the procedure described
above terminates after a finite number of steps. At this point, every
partition $\pi_i(F')$ belongs to the corresponding set $\Pi_i$, that
is, the solution $F'$ conforms to $\Pi$.
\end{proof}
\fi

\iffull
\section{Algorithm for series-parallel graphs}\label{sec:algor-seri-parall}

\iffull
A series-parallel graph is a graph that can be built using series and
parallel composition. 
\fi
Formally, a {\em series-parallel graph} $G(x,y)$ with distinguished vertices
$x,y$ is an undirected graph that can be constructed using the
following rules:
\begin{itemize}
\item An edge $xy$ is a series-parallel graph.
\item If $G_1(x_1,y_1)$ and $G_2(x_2,y_2)$ are series-parallel graphs, then the graph $G(x,y)$
  obtained by identifying $x_1$ with $x_2$ and $y_1$ with $y_2$ is a
  series-parallel graph with distinguished vertices $x:=x_1=x_2$ and $y:=y_1=y_2$
  {\em (parallel connection).}
\item If $G_1(x_1,y_1)$ and $G_2(x_2,y_2)$ are series-parallel
  graphs, then the graph $G(x,y)$ obtained by identifying $y_1$ with $x_2$
  is a series-parallel graph with distinguished vertices $x:=x_1$ and
  $y:=y_2$ {\em (series connection).}
\end{itemize}

We prove Theorem~\ref{th:mainalg} in this section by constructing a
polynomial-time algorithm to solve \prob{Steiner Forest} on
series-parallel graphs. It is well-known that the treewidth of a
graph is at most 2 if and only if it is a subgraph of a series
parallel graph\cite{Bodlaender98}. Since setting the length of an
edge to $\infty$ is essentially the same as deleting the edge, it
follows that \prob{Steiner Forest} can be solved in polynomial time
on graphs with treewidth at most 2.

Let $(G,\DD)$ be an instance of \prob{Steiner Forest} where $G$ is
series parallel.  For $i=1,\dots, m$, denote by $G_i(x_i,y_i)$ all the
intermediary graphs appearing in the series-parallel construction of
$G$. We assume that these graphs are ordered such that $G=G_m$ and if
$G_i$ is obtained from $G_{j_1}$ and $G_{j_2}$, then $j_1,j_2<i$. Let
$\DD_i\subseteq \DD$ contain those pairs $\{u,v\}$ where both vertices
are in $V(G_i)$.  Let $A_i$ be those vertices $v\in V(G_i)$ for which
there exists a pair $\{v,u\}\in \DD$ with $u\not\in V(G_i)$ (note that
$A_m=\emptyset$ and $\DD_m=\DD$).  For every $G_i$, we define two
integer values $a_i,b_i$ and a function $f_i$:
\begin{itemize}
\item Let $a_i$ be the minimum cost of a solution $F$ of the
  instance $(G_i,\DD_i)$ with the additional requirements that $x_i$ and
  $y_i$ are connected in $F$ and every vertex in $A_i$ is in the same
  component as $x_i$ and $y_i$.
\item Let $G'_i$ be the graph obtained from $G_i$ by identifying
  vertices $x_i$ and $y_i$. Let $b_i$ be the minimum cost of a solution $F$ of
  the instance $(G'_i,\DD_i)$ with the additional requirement that
  $A_i$ is in the same component as $x_i=y_i$.
\item For every $S\subseteq A_i$, let $f_i(S)$ be the minimum cost of a solution $F$ of the
  instance $(G_i,\DD_i)$ with the additional requirements that $x_i$
  and $y_i$ are not connected, every
  $v\in S$ is in the same component as $x_i$, and every $v\in A_i\setminus S$ is
  in the same component as $y_i$. (If there is no such $F$, then $f_i(S)=\infty$.)
\end{itemize}
The main combinatorial property that allows us to solve the problem in
polynomial time is that the functions $f_i$ are submodular. We prove
something stronger: the functions $f_i$ can be represented in a
compact way as the cut function of certain directed graphs.

If $D$ is a directed graph length on the edges and $X\subseteq V(D)$,
then $\delta_D(X)$ denotes the total length of the edges leaving $X$.
For $X,Y\subseteq V(D)$, we
denote by $\lambda_D(A,B)$ the minimum length of a directed cut that
separates $A$ from $B$, i.e., the minimum of $\delta_D(X)$, taken over
all $A\subseteq X \subseteq V(D)\setminus B$ (if $A\cap B\neq \emptyset$, then $\lambda_D(A,B)$ is defined
to be $\infty$).

\begin{definition}\label{def:represent}
Let $D_i$ be a directed graph with nonnegative lengths. Let $s_i$ and
$t_i$ be two distinguished vertices and let $A_i$ be a subset of
vertices of $D_i$. We say that $(D_i,s_i,t_i,A_i)$ {\em represents} $f_i$ if
$f_i(S)=\lambda_{D_i}(S\cup \{s_i\},(A_i\setminus S)\cup \{t_i\})$ for every
$S\subseteq A_i$. If $s_i$, $t_i$, $A_i$ are clear from the context,
then we simply say that $D_i$ represents $f_i$.
\end{definition}

A function $f$ defined on the subsets of a ground set $U$ is {\em
  submodular} if
\[
f(X)+f(Y)\ge f(X\cap Y)+f(X\cup Y)
\]
holds for every $X,Y\subseteq U$. For example, it is well known that
$\delta_G(X)$ is a submodular function on the subsets of $V(G)$.
Submodularity is a powerful unifying concept of combinatorial
optimization: classical results on flows, cuts, matchings, and
matroids can be considered as consequences of submodularity. The
following (quite standard) proposition shows that if a function  can be
represented in the sense of Definition~\ref{def:represent}, then the function
is submodular. In the proof of Theorem~\ref{th:mainalg}, we show that
every function $f_i$ can be represented by a directed graph, thus
it follows that every $f_i$ is submodular. Although we do not use this
observation directly in the paper, it explains in some sense why the
problem is polynomial-time solvable.

\begin{proposition}
If a function $f_i:2^{A_i}\to \mathbb{R}^+$ can be represented by
$(D_i,s_i,t_i,A_i)$ (in the sense of Definition~\ref{def:represent}), then  $f_i$ is submodular.
\end{proposition}
\begin{proof}
Let $X,Y\subseteq A_i$ be arbitrary sets. Since $D_i$ represents
$f_i$, there exist appropriate sets $X'$ and $Y'$ with
$\delta_{D_i}(X')=f_i(X)$ and $\delta_{D_i}(Y')=f_i(Y)$. Now we have
\begin{align}
f_i(X)+f_i(Y)=\delta_{D_i}(X')+\delta_{D_i}(Y')\ge
\delta_{D_i}(X'\cap Y')+\delta_{D_i}(X'\cup Y'),  \label{eqn:represent-prop}
\end{align}
where the inequality follows from the submodularity of
$\delta_{D_i}$. Observe that $(X\cap Y)\cup \{s\}\subseteq X'\cap Y'$
and $X'\cap Y'\subseteq V(D_i)\setminus ((A_i\setminus X)\cup
(A_i\setminus Y)\cup\{t\})$. Thus we have $f_i(X\cap Y)\le
\delta_{D_i}(X'\cap Y')$. In a similar way, $f_i(X\cup Y)\le
\delta_{D_i}(X'\cup Y')$.
Together with Inequality~\eqref{eqn:represent-prop} obtained above,
  this proves that $f_i(X)+f_i(Y)\ge f_i(X\cap Y)+f_i(X\cup Y)$.
\end{proof}

\begin{proof}[Proof of Theorem~\ref{th:mainalg}] We assume that in
the given instance of \prob{Steiner Forest} each vertex appears only
in at most one pair of $\DD$. To achieve this, if a vertex $v$
appears in $k>1$ pairs, then we subdivide an arbitrary edge incident
to $v$ by $k-1$ new vertices such that each of the $k-1$ edges on
the path formed by $v$ and the new vertices has length 0. Replacing
vertex $v$ in a pair by any of the new vertices does not change the
problem.

For every $i=1,\dots,m$, we compute the values $a_i$, $b_i$, and a
representation $D_i$ of $f_i$. In the optimum solution $F$ for the
instance $(G_m,\DD)$, vertices
$x_m$ and $y_m$ are either connected or not. Thus the cost of the
optimum solutions is the minimum of $a_m$ and $f_m(\emptyset)$ (recall
that $A_m=\emptyset$). The value of $f_m(\emptyset)$ can be easily
determined by computing the minimum cost $s$-$t$ cut in $D_m$.

If $G_i$ is a single edge $e$, then $a_i$ and $b_i$ are trivial to
determine: $a_i$ is the length of $e$ and $b_i=0$. The directed graph
$D_i$ representing $f_i$ can be obtained from $G_i$ by renaming $x_i$ to $s_i$, renaming
$y_i$ to $t_i$, and either removing the edge $e$ (if $\DD_i=\emptyset$)
or replacing $e$ with a directed edge $\overrightarrow{s_it_i}$ of length
$\infty$ (if  $\{u,v\}\in \DD_i$).

If $G_i$ is not a single edge, then it is constructed from some
$G_{j_1}$ and $G_{j_2}$ either by series or parallel
connection. Suppose that $a_{j_\myell}$, $b_{j_\myell}$, and $D_{j_\myell}$
for $\myell=1,2$ are already known. We show how to compute $a_i$, $b_i$,
and $D_i$ in this case.

\textbf{Parallel connection.} Suppose that $G_i$ is obtained from
$G_{j_1}$ and $G_{j_2}$ by parallel connection. It is easy to see that
$a_i=\min\{a_{j_1}+b_{j_2},b_{j_1}+a_{j_2}\}$ and
$b_i=b_{j_1}+b_{j_2}$. To obtain $D_i$, we join $D_{j_1}$ and
$D_{j_2}$ by identifying $s_{j_1}$ with $s_{j_2}$ (call it $s_i$)
and by identifying $t_{j_1}$ with $t_{j_2}$ (call it
$t_i$). Furthermore, for every $\{u,v\}\in \DD_i\setminus
\{\DD_{j_1}\cup \DD_{j_2}\}$, we add directed edges
$\overrightarrow{uv}$ and $\overrightarrow{vu}$ with length $\infty$.

To see that $D_i$ represents $f_i$, suppose that $F$ is the subgraph
that realizes the value $f_i(S)$ for some $S\subseteq A_i$. We have to
show that there is an appropriate $X\subseteq V(D_i)$ certifying
$\lambda_{D_i}(S\cup\{s\},(A_i\setminus S)\cup \{t\})\le \ell(F)$. The
graph $F$ is the edge disjoint union of two graphs $F_1\subseteq
G_{j_1}$ and $F_2\subseteq G_{j_2}$. For $\myell=1,2$, let
$S^\myell\subseteq A_{j_\myell}$ be the set of those vertices that are
connected to $x_{j_\myell}$ in $F_\myell$, it is clear that $F_\myell$
connects $A_{j_\myell}\setminus S^\myell$ to $y_{j_\myell}$. Since $F_\myell$
does not connect $x_i$ and $y_i$, we have that $\ell(F_\myell)\ge
f_{j_\myell}(S^\myell)$. Since $D_{j_{\myell}}$ represents $f_{j_{\myell}}$,
there is a set $X_\myell$ in $D_{j_\myell}$ with $S^\myell\cup \{
s_{j_\myell}\} \subseteq X_\myell \subseteq V(D_i)\setminus
((A_{j_{\myell}}\setminus S^\myell)\cup \{t_{j_\myell}\})$, and
$\delta_{D_{j_\myell}}(X_\myell)=f_{j_\myell}(S^\myell)$. We show that
$\delta_{D_i}(X_1\cup
X_2)=\delta_{D_{j_1}}(X_1)+\delta_{D_{j_2}}(X_2)$. Since $D_i$ is
obtained from joining $D_{j_1}$ and $D_{j_2}$, the only thing that has
to be verified is that the edges with infinite length added after the
join cannot leave $X_1\cup X_2$. Suppose that there is such an edge
$\overrightarrow{uv}$, assume without loss of generality that $u\in
X_1$ and $v\in V(D_{j_2})\setminus X_2$. This means that $u\in S^1$
and $v\not \in S^2$. Thus $F$ connects $u$ with $x_i$ and $v$ with
$y_i$, implying that $F$ does not connect $u$ and $v$. However
$\{u,v\}\in \DD_i$, a contradiction. Therefore, for the set $X:=X_1\cup
X_2$, we have
\begin{multline*}
\delta_{D_i}(X)=\delta_{D_{j_1}}(X_1)+\delta_{D_{j_2}}(X_2)=f_{j_1}(S^1)+
f_{j_2}(S^2) \le \ell(F_1)+\ell(F_2)=\ell(F)=f_i(S),
\end{multline*}
proving the existence of the required $X$.

Suppose now that for some $S\subseteq A_i$, there is a set $X$ with $S\cup \{ s_i\}
\subseteq X \subseteq V(D_i)\setminus
((A_{i}\setminus S)\cup \{t_{i}\})$. We have to show that
$\delta_{D_i}(X)\ge f_i(S)$. For $\myell=1,2$, let $X_\myell=X \cap
V(D_{j_\myell})$ and $S^\myell=A_{j_\myell}\cap X_\myell$. Since $D_{j_\myell}$
represents $f_{j_\myell}$, we have that $\delta_{D_{j_\myell}}(X_\myell)\ge
f_{j_\myell}(S^\myell)$. Let $F_\myell$ be a subgraph of
$G_{j_\myell}$ realizing $f_{j_\myell}(S^\myell)$. Let $F=F_1\cup
F_2$; we show that $\ell(F)\ge f_i(A_i)$, since $F$ satisfies all the
requirements in the definition of $f_i(A_i)$. It is clear that $F$
does not connect $x_i$ and $y_i$. Consider a pair $\{u,v\}\in
\DD_i$. If $\{u,v\}\in \DD_{j_\myell}$, then $F$ connects $u$ and
$v$. Otherwise, let $\{u,v\}\in \DD_i\setminus \{\DD_{j_1}\cup
\DD_{j_2}\}$. Suppose that $F$
does not connect $u$ and $v$, without loss of generality, suppose that
$u\in X_1$ and $v\not \in X_2$. This means that there is an edge
$\overrightarrow{uv}$ of length $\infty$ in $D_i$ and
$\delta_{D_i}(X)=\infty\ge f_i(S)$ follows. Thus we proved $\ell(F)\ge f_i(A_i)$ and we
have
\begin{multline*}
\delta_{D_i}(X)=\delta_{D_{j_1}}(X_1)+\delta_{D_{j_2}}(X_2) \ge
f_{j_1}(S^1)+f_{j_2}(S^2)=
\ell(F_1)+\ell(F_2)=\ell(F) \ge f_i(S),
\end{multline*}
what we had to show.

\textbf{Series connection.}  Suppose that $G_i$ is obtained from
$G_{j_1}$ and $G_{j_2}$ by series connection and let
$\mu:=y_{j_1}=x_{j_2}$ be the middle vertex. It is easy to see that
$a_i=a_{j_1}+a_{j_2}$ (vertex $\mu$ has to be connected to both $x_i$
and $y_i$). To compute $b_i$, we argue as follows. Denote by
$G^R_{j_2}$ the graph obtained from $G_{j_2}$ by swapping the names of
distinguished vertices $x_{j_2}$ and $y_{j_2}$. Observe that the graph
$G'_i$ in the definition of $b_i$ arises as the parallel connection of
$G_{j_1}$ and $G^R_{j_2}$. It is easy to see that $a^R_{j_2}$,
$b^R_{j_2}$, and $f^R_{j_2}$ corresponding to $G^R_{j_2}$ can be
defined as $a^R_{j_2}=a_{j_2}$, $b^R_{j_2}=b_{j_2}$, and
$f^R_{j_2}(S)=f_{j_2}(A_{j_2}\setminus S)$. Furthermore, if $D_{j_2}$
represents $f_{j_2}$, then the graph $D^R_{j_2}$ obtained from
$D_{j_2}$ by swapping the names of $s_{j_2}$ and $t_{j_2}$ represents
$f^R_{j_2}$. Thus we have everything at our disposal to construct a
directed graph $D'_i$ that represents the function $f'_i$ corresponding
to the parallel connection of $G_{j_1}$ and $G^R_{j_2}$. Now it is
easy to see that $b_i=f'_i(A_i)$: graph $G'_i$ is isomorphic to the
parallel connection of $G_{j_1}$ and $G^R_{j_2}$ and the definition of
$b_i$ requires that $A_i$ is connected to $x_i=y_i$. The value of
$f'_i(A_i)$ can be determined by a simple minimum cut computation in
$D'_i$.

Let $T_1\subseteq A_{j_1}$ contain those vertices $v$ for
which there exists a pair $\{v,u\}\in \DD_i$ with $u\in A_{j_2}$ and
let $T_2\subseteq A_{j_2}$ contain those vertices $v$ for which there
exists a pair $\{v,u\}\in \DD_i$ with $u\in A_{j_1}$. Observe that
$A_i=(A_{j_1}\setminus T_1)\cup (A_{j_2}\setminus T_2)$ (here we are
using the fact that each vertex is contained in at most one pair, thus
a $v\in T_1$ cannot be part of any pair $\{v,u\}$ with $u\not\in
V(D_i)$). To construct $D_i$, we connect $D_{j_1}$
and $D_{j_2}$ with an edge $\overrightarrow{t_{j_1}s_{j_2}}$ of length
0 and set $s_i:=s_{j_1}$ and $t_i:=t_{j_2}$. Furthermore, we introduce
two new vertices $\gamma_1,\gamma_2$ and add the following edges (see Figure~\ref{fig:series}):
\begin{itemize}
\item $\overrightarrow{s_{j_1}\gamma_1}$ with length $a_{j_2}$,
\item $\overrightarrow{\gamma_1\gamma_2}$ with length $f_{j_1}(A_{j_1}\setminus
T_1)+f_{j_2}(T_{2})$,
\item $\overrightarrow{\gamma_2t_{j_2}}$ with length $a_{j_1}$,
\item $\overrightarrow{\gamma_2\gamma_1}$ with length $\infty$,
\item $\overrightarrow{\gamma_1v}$ with length $\infty$ for every $v\in V(D_{j_1})$,

\item $\overrightarrow{v\gamma_2}$ with length $\infty$ for every $v\in V(D_{j_2})$,
\item $\overrightarrow{v\gamma_1}$ with length $\infty$ for every $v\in
  T_1$, and
\item $\overrightarrow{\gamma_2v}$ with length $\infty$ for every $v\in T_2$.

\end{itemize}

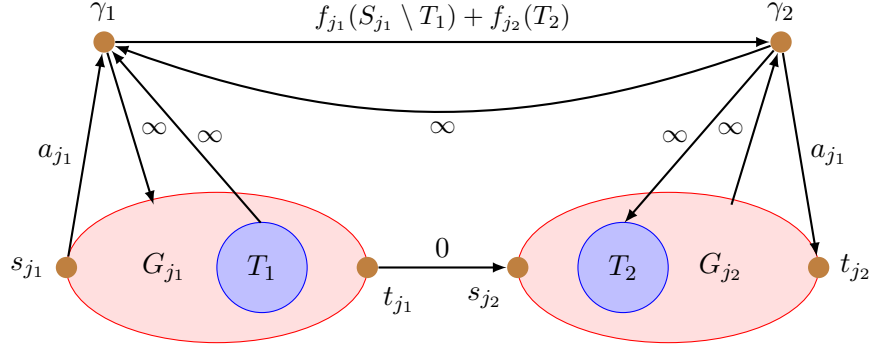
\begin{figure}
\begin{center}
\begin{tikzpicture}
\tikzstyle{elp}=[draw=red,fill=pink!50!white]
\tikzstyle{cir}=[draw=blue,fill=blue!25!white]
\tikzstyle{nod}=[circle,minimum size=1mm,inner sep=1mm,fill=brown]x
\tikzstyle{ar}=[thick,-latex]
\draw[elp] (1,0) ellipse (2cm and 1cm);
\draw[cir] (1,0) ++(.6,0) circle (.6cm);
\draw (.3,0) node {$G_{j_1}$};
\draw (1.6,0) node {$T_1$};
\draw (-1,0) node[nod,label=left:$s_{j_1}$] (N11) {};
\draw (3,0) node[nod,label=-60:$t_{j_1}$] (N12) {};
\draw (-.5,3) node[nod,label=above:$\gamma_1$] (N13) {};
\draw (.2,.7) node (N14) {};
\draw (1.7,.46) node (N15) {};
\draw[elp] (7,0) ellipse (2cm and 1cm);
\draw[cir] (7,0) ++(-.6,0) circle (.6cm);
\draw (8-.3,0) node {$G_{j_2}$};
\draw (8-1.6,0) node {$T_2$};
\draw (8+1,0) node[nod,label=right:$t_{j_2}$] (N21) {};
\draw (8-3,0) node[nod,label=-120:$s_{j_2}$] (N22) {};
\draw (8+.5,3) node[nod,label=above:$\gamma_2$] (N23) {};
\draw (8-.2,.7) node (N24) {};
\draw (8-1.7,.46) node (N25) {};
\draw[ar] (N11) -- node[left] {$a_{j_1}$} (N13);
\draw[ar] (N13) -- node[right] {\small$\infty$} (N14);
\draw[ar] (N15) -- node[right] {\small$\infty$} (N13);
\draw[ar] (N12) -- node[above] {$0$} (N22);
\draw[ar] (N23) -- node[right] {$a_{j_1}$} (N21);
\draw[ar] (N24) -- node[left] {\small$\infty$} (N23);
\draw[ar] (N23) -- node[left] {\small$\infty$} (N25);
\draw[ar] (N23) to [bend left=20] node[below] {\small$\infty$} (N13);
\draw[ar] (N13) -- node[above] {\small$f_{j_1}(S_{j_1}\setminus T_1)+f_{j_2}(T_2)$} (N23);
\end{tikzpicture}
\end{center}
\caption{Construction of $D_i$ in a series connection.}\label{fig:series}
\end{figure}

Suppose that $F$ is the subgraph that realizes the value $f_i(S)$ for
some $S\subseteq A_i$; we have to show that $D_i$ has an appropriate
cut with value $f_i(S)$. Subgraph $F$ is the edge-disjoint union of
subgraphs $F_1\subseteq G_{j_1}$ and $F_2\subseteq G_{j_2}$. We
consider 3 cases: in subgraph $F$, vertex $\mu$ is either connected to
neither $x_{j_1}$ nor $y_{j_2}$, connected only to $x_{j_1}$, or
connected only to $y_{j_2}$.

Case 1: $\mu$ is connected to neither $x_{j_1}$ nor
$y_{j_2}$. In this case, vertices of $A_{j_1}$ are not connected to
$y_i$ and vertices of $A_{j_2}$ are not connected to $x_i$, hence
$S=A_i\cap A_{j_1}=A_{j_1}\setminus T_1$ is the only possibility. Furthermore, $F$
connects both $T_1$ and $T_2$ to $\mu$. It follows that
$\ell(F)=\ell(F_1)+\ell(F_2)\ge f_{j_1}(A_{j_1}\setminus T_1)+f_{j_2}(T_2)$. Set $X=V(D_i)\cup
\{\gamma_1\}$: now we have $\delta_{D_i}(X)=f_{j_1}(A_{j_1}\setminus T_1)+f_{j_2}(T_2)\le \ell(F)$, $X$ contains
$(A_{j_1}\setminus T_1)\cup\{s_i\}$, and is disjoint from
$(A_{j_2}\setminus T_2)\cup
\{t_i\}$.

Case 2: $\mu$ is connected only to $x_{i}$. This is only possible if
$A_{j_1}\setminus T_1 \subseteq S$. Clearly, $\ell(F_1)\ge a_{j_1}$.
Subgraph $F_2$ has to connect every vertex in $(S\cap A_{j_2})\cup
T_2$ to $x_{j_2}$ and every vertex $A_{i}\setminus
S=A_{j_2}\setminus(S\cup T_2)$ to $y_{j_2}$. Hence $\ell(F_2)\ge
f_{j_2}((S\cap A_{j_2})\cup T_2)$ and let $X_2\subseteq V(D_{j_2})$ be
the corresponding cut in $D_{j_2}$. Set $X:=X_2\cup V(D_{j_1})\cup
\{\gamma_1,\gamma_2\}$, we have $\delta_{D_i}(X)=f_{j_2}((S\cap A_{j_2})\cup
T_2)+a_{j_1}\le \ell(F_2)+a_{j_1} \le \ell(F)$ (note that no edge with
infinite length leaves $X$ since $T_2\subseteq X_2$). As $X$ contains
$S$ and contains none of the vertices in $A_i\setminus S$, we proved
the existence of the required cut.

Case 3: $\mu$ is connected only to $y_{i}$. Similar to
case 2.

Suppose now that for some $S\subseteq A_i$, there is a set $X\subseteq
V(D_i)$ with $S\cup \{ s_i\} \subseteq X \subseteq V(D_i) \setminus ((A_{i}\setminus
S)\cup \{t_{i}\})$. We have to show that $\delta_{D_i}(X)\ge
f_i(S)$. If $\delta_{D_i}(X)=\infty$, then there is nothing to show.
  In particular, because of the edge $\overrightarrow{\gamma_2\gamma_1}$, we are
  trivially done if $\gamma_2\in X$ and $\gamma_1\not\in X$. Thus we have to
  consider only 3 cases depending which of $\gamma_1,\gamma_2$ are contained in
  $X$.

  Case 1: $\gamma_1\in X$, $\gamma_2\not\in X$. In this case, the edges having
  length $\infty$ ensure that $V(D_{j_1})\subseteq X$ and
  $V(D_{j_2})\cap X=\emptyset$, thus
  $\delta_{D_i}(X)=\ell(\overrightarrow{\gamma_1\gamma_2})+\ell(\overrightarrow{t_{j_1}s_{j_2}})=f_{j_1}(A_{j_1}\setminus T_1)+f_{j_2}(T_{j_2})$.
We also have $S=A_{j_1}\setminus T_1$. Now it is easy to see that
$f_i(S)\le f_{j_1}(A_{j_1}\setminus T_1)+f_{j_2}(T_{j_2})$: taking
the union of some $F_1$ realizing $f_{j_1}(A_{j_1}\setminus T_1)$ and
some $F_2$ realizing $f_{j_2}(T_{j_2})$, we get a subset $F$ realizing
$f_i(S)$.

Case 2: $\gamma_1,\gamma_2\in X$. The edges starting from $\gamma_1$ and having
length $\infty$ ensure that $V(D_{j_1})\subseteq X$. Furthermore,
$\gamma_2\in X$ ensures that $T_2\subseteq X$. Let $X_2:=X\cap
V(D_{j_2})$, we have $X_2\cap A_{j_2}=T_2\cup (S\cap A_{j_2})$, which
implies $\delta_{D_{j_2}}(X_2)\ge f_{j_2}(T_2\cup (S\cap A_{j_2}))$.
Observe that $\delta_{D_i}(X)=a_{j_1}+\delta_{D_{j_2}}(X_2)$ (the term
$a_{j_1}$ comes from the edge $\overrightarrow{\gamma_2t_{i}}$).
Let $F_1$ be a subset of $G_{j_1}$ realizing $a_{j_1}$ and let
$F_2$ be a subset of $G_{j_2}$ realizing $f_{j_2}(T_2\cup (S\cap
A_{j_2}))$. Let $F:=F_1\cup F_2$, note that $F$ connects vertices $S$
with $x_i$, vertices $A_i\setminus S$ with $y_i$, and vertices in
$T_1\cup T_2$ with $\mu$. Thus $f_{i}(S)\le \ell(F)=a_{j_1}+
f_{j_2}(T_2\cup (S\cap A_{j_2}))=\delta_{D_i}(X)$, what we had to
show.

Case 3: $\gamma_1,\gamma_2\not\in X$. Similar to Case 2.
\end{proof}

\section{Hardness for treewidth 3}\label{sec:hardness-treewidth-3}

In this section, we show that \prob{Steiner Forest} is \textup{NP}-hard on
graphs with treewidth at most 3. Very recently, this has been proved
independently by Gassner~\cite{Gassner2009}, but our compact proof perhaps better
explains what the reason is for the sharp difference between the
series-parallel and the treewidth 3 cases.

Consider the graph in Figure~\ref{fig:hardness} and let us define the function $f$
analogously to the function $f_i$ in
Section~\ref{sec:algor-seri-parall}: for every $S\subseteq \{1,2,3\}$,
let $f(S)$ be the minimum cost of a subgraph $F$ where $x$ and $y$ are
not connected,  $t_i$ is connected to $x$ for every $i\in S$, and
$t_i$ is connected to $y$ for every $i\in  \{1,2,3\}\setminus S$; if
there is no such subgraph $F$, then let $f(S)=\infty$. It is easy to
see that $f(\{1,2\})=f(\{2,3\})=f(\{1,2,3\})$, while
$f(\{2\})=\infty$. Thus, unlike in the case of series-parallel graphs,
this function is not submodular.
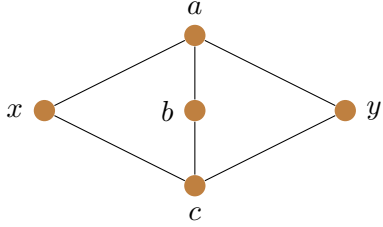
\begin{figure}
\begin{center}
\begin{tikzpicture}
\tikzstyle{nod}=[circle,minimum size=1mm,inner sep=1mm,fill=brown]x
\draw (-2,0) node[nod,label=left:$x$] (Nx) {};
\draw (2,0) node[nod,label=right:$y$] (Ny) {};
\draw (0,1) node[nod,label=above:$a$] (Na) {};
\draw (0,-1) node[nod,label=below:$c$] (Nc) {};
\draw (0,0) node[nod,label=left:$b$] (Nb) {};
\draw (Nx) -- (Na);
\draw (Nx) -- (Nc);
\draw (Ny) -- (Na);
\draw (Ny) -- (Nc);
\draw (Na) -- (Nb);
\draw (Nb) -- (Nc);
\end{tikzpicture}
\caption{The graph used in the proof Theorem~\ref{th:boundedhardness}.}\label{fig:hardness}
\end{center}
\end{figure}

We use the properties of the function $f$ defined in the previous
paragraph to obtain a hardness proof in a more or less ``automatic''
way. 
Let us define the Boolean relation $R(a,b,c):=(a=c)\vee (b=c)$.
Observe that for any $S\subseteq \{1,2,3\}$, we have $f(S)=3$ if
$R(1\in S,2\in S,3\in S)=1$ and $f(S)=\infty$ otherwise (here $1\in S$
means the Boolean variable that is 1 if and only if $1\in S$).
 An
{\em R-formula} is a conjunction of clauses, where each clause is the
relation $R$ applied to some Boolean variables or to the constants 0 and 1,
e.g., $R(x_1,0,x_4)\wedge R(0,x_2,x_1) \wedge R(x_3,x_2,1)$. In the
\prob{$R$-Sat} problem, the input is an $R$-formula and it has to be
decided whether the formula has a satisfying assignment.

\begin{lemma}
\prob{$R$-Sat} is \textup{NP}-complete.
\end{lemma}

\begin{proof}
  Readers familiar with Schaefer's Dichotomy Theorem (more precisely,
  the version allowing constants \cite[Lemma 4.1]{MR80d:68058}) can
  easily see that $R$-SAT is NP-complete.  It is easy to verify that
  the relation $R$ is neither weakly positive, weakly negative,
  affine, nor bijunctive. Thus the result of Schaefer immediately
  implies that \prob{$R$-Sat} is NP-complete.

  For completeness, we give a simple self-contained proof here.  The
  reduction is from \prob{Not-All-Equal 3SAT}\footnote{In
    \prob{Not-All-Equal 3SAT}, or \prob{NAE-3SAT} for short, we are
    given a 3SAT instance with the extra restriction that a clause is
    not satisfied if \emph{all} the literals in a clause are true.
    Similarly to \prob{3SAT}, the clause is not satisfied if all the
    literals are false, either.  Thus, the literals in each clause
    have to take both true and false values.}, which is known to be
  NP-complete even if there are no negated literals
  \cite{MR80d:68058}. Given a NAE-SAT formula, we replace each clause
  as follows. For each clause $\text{NAE}(a,b,c)$, we introduce a new
  variable $d$ and create the clauses $R(a,b,d)\wedge R(c,d,0) \wedge
  R(c,d,1)$. If $a=b=c$, then it is not possible that all three
  clauses are simultaneously satisfied (observe that the second and
  third clauses force $c\neq d$). On the other hand, if $a$, $b$, $c$
  do not have the same value, then all three clauses can be satisfied
  by an appropriate choice of $d$. Thus the transformation from
  NAE-SAT to \prob{$R$-Sat} preserves satisfiability.
\end{proof}

The main idea of the following proof is that we can simulate
arbitrarily many $R$-relations by joining in parallel copies of the
graph shown in Figure~\ref{fig:hardness}.
\begin{theorem}\label{th:boundedhardness}
\prob{Steiner Forest} is \textup{NP}-hard on planar graphs with treewidth at
most 3.
\end{theorem}
\begin{proof}
  The proof is by reduction from \prob{$R$-Sat}. Let $\phi$ be an
  $R$-formula having $n$ variables and $m$ clauses. We start the
  construction of the graph $G$ by introducing two vertices $v_0$ and
  $v_1$. For each variable $x_i$ of $\phi$, we introduce a vertex
  $x_i$ and connect it to both $v_0$ and $v_1$. We introduce 3 new
  vertices $a_i$, $b_i$, $c_i$ corresponding to the $i$-th clause.
  Vertices $a_i$ and $b_i$ are connected to both $v_0$ and $v_1$,
  while $c_i$ is adjacent only to $a_i$ and $b_i$. If the $i$-th
  clause is $R(x_{i_1},x_{i_2},x_{i_3})$, then we add the 3 pairs
  $\{x_{i_1},a_i\}$, $\{x_{i_2},b_i\}$, $\{x_{i_3},c_i\}$ to $\DD$. If
  the clause contains constants, then we use the vertices $v_0$ and
  $v_1$ instead of the vertices $x_{i_1}$, $x_{i_2}$, $x_{i_3}$, e.g.,
  the clause $R(0,x_{i_2},1)$ yields the pairs $\{v_0,a_i\}$,
  $\{x_{i_2},b_i\}$, $\{v_1,c_i\}$. The length of every edge is 1.
  This completes the description of the graph $G$ and the set of pairs
  $\DD$.

  We claim that the constructed instance of \prob{Steiner Forest}
  has a solution with $n+3m$ edges if and only if the $R$-formula
  $\phi$ is satisfiable. Suppose that $\phi$ has a satisfying
  assignment $f$. We construct the $F$ as follows. If $f(x_i)=1$, then
  let us add edge $x_iv_1$ to $F$; if $f(x_i)=0$, then
  let us add edge $x_iv_0$ to $F$. For each clause, we add 3 edges to
  $F$. Suppose that the $i$-th clause is
  $R(x_{i_1},x_{i_2},x_{i_3})$. We add one of $a_iv_0$ or $a_iv_1$ to
  $F$ depending on the value of $f(x_{i_1})$ and we add one of
  $b_iv_0$ or $b_iv_1$ to $F$ depending on the value of
  $f(x_{i_2})$. Since the clause is satisfied, either
  $f(x_{i_3})=f(x_{i_1})$ or $f(x_{i_3})=f(x_{i_2})$; we add $c_ia_i$
  or $c_ib_i$ to $F$, respectively (if $f(x_{i_3})$ is equal to both,
  then the choice is arbitrary). We proceed in an analogous manner for
  clauses containing constants. It is easy to verify that each pair
  is in the same connected component of $F$.

  Suppose now that there is a solution $F$ with cost $n+3m$. At
  least one edge is incident with each vertex $x_i$, since it cannot
  be isolated in $F$. Each vertex $a_i$, $b_i$, $c_i$ has to be
  connected to either $v_0$ or $v_1$, hence at least $3$ edges of $F$
  are incident with these 3 vertices. As $F$ has $n+3m$ edges, it
  follows that exactly one edge is incident with each $x_i$, and hence
  exactly 3 edges are incident with the set $\{a_i,b_i,c_i\}$. It
  follows that $v_0$ and $v_1$ are not connected in $F$. Define an
  assignment of $\phi$ by setting $f(x_i)=0$ if and only if vertex
  $x_i$ is in the same component of $F$ as $v_0$. To verify that a
  clause $R(x_{i_1},x_{i_2},x_{i_3})$ is satisfied, observe that $c_i$
  is in the same component of $F$ as either $a_i$ or $b_i$. If $c_i$
  is in the same component as, say, $a_i$, then this component also
  contains $x_{i_3}$ and $x_{i_1}$, implying $f(x_{i_3})=f(x_{i_1})$
  as required.
\end{proof}

\fi

\bibliographystyle{alpha}
\bibliography{main}

\end{document}